\documentclass[leqno,11pt]{amsart}
\usepackage{amsmath,amssymb,amsfonts,amsthm}
\usepackage{amsfonts}
\usepackage{graphics}
\usepackage{graphicx}
\usepackage{subfigure}

\newtheorem{theorem}{Theorem}

\newtheorem{example}[theorem]{Example}

\newtheorem{proposition}[theorem]{Proposition}
\newtheorem{remark}[theorem]{Remark}

\begin{document}
\title[]{ Pricing  Bitcoin Derivatives under Jump-Diffusion Models}
\author[]{Pablo Olivares, Ryerson University}
\address{}

\begin{abstract}
In recent years cryptocurrency trading has captured the attention of practitioners and academics. The volume of the exchange with standard currencies has knoew a dramatic increasing of late. This paper addresses to the need of models describing a bitcoin-US dollar exchange dynamic and their use to evaluate European option having bitcoin as underlying asset.
\end{abstract}

\maketitle
\section{Introduction}
In recent years cryptocurrency trading has captured the attention of practitioners and academics. The volume in the exchange of the former with standard currencies has known a dramatic increasing of late. Due to the special circumstances in which the mining of the cryptocurrencies  take place and its lack of transparency, the dynamics of the rate of change is characterized by a high volatility and large random oscillations upon time. This situation introduces an extra degree of difficulty in the modeling of exchange data.\\
On the other hand, there is, an  informal but emerging market for derivatives based on cryptocurrencies. Evaluation of future contracts have recently appeared on some web sites. The market for more complex derivatives is at an incipient stage. Moreover, to our knowledge  the pricing of the latter has not been analyzed. \\
This paper addresses to the need of evaluating the latter. To this end, we propose a model for the dynamic of the exchange rates based on a mean-reverting exponential Levy process with jump-diffusion log-returns. We study empirical properties of the probability laws in bitcoin-US dollar exchanges and correlation, as well as parameter estimation from three different perspectives.  Next, we study the pricing of European options adapting well-known Fast Fourier Transform  techniques (FFT) for Levy process established in Car and Madan (1999) to this context.\\
The organization of the paper is the following:\\
In  section 2 we introduce the model, the risk-neutral setting and compute the characteristic function of the log-returns of the exchanges. In section 3 we specify these results for Merton(1976) and Kou(2002) jump-diffusion models. In section 4 we study empirical behavior of bitcoin-US dollar exchange data and parameter estimation. Finally, in section 5 we outline the pricing method, while in section 6 we conclude.
\section{Modeling bitcoin-US dollar exchange dynamic}
 Let  $(\Omega ,\mathcal{A}, (\mathcal{F}_{t})_{t \geq 0}, P)$ be a filtered probability space verifying the usual conditions. We denote  by $\mathcal{Q}$ an equivalent martingale measure(EMM) and by $E_{\mathcal{Q}}$ and $\varphi_X$ respectively the expected value and characteristic function  of a random variable $X$ under $\mathcal{Q}$. Furthermore, the function $l_V(u)=\frac{1}{t} \log \varphi_{V_{t}}(-iu)$ is the Laplace exponent of a Levy process  $(V_t)_{t \geq 0}$ defined on the space above.
  %Notice that, when $(X)_{t \geq 0}$ is a Levy process $l_X(u,s,t)=(t-s)l_X(u)=(t-s)\Psi_X(-iu)$, where $\Psi_X$ is its characteristic exponent.
   The symbol $\hat{f}$ denotes the Fourier transform of a function $f$, while $D^kf(u)$ or  $f^{(k)}$ denote its k-th derivative with respect to $u$. We set $Df:=D^1f$. For a random variable $X_t$ the expression $\tilde{X}_t= e^{-rt}X_t$ denotes its discounted value with respect to a contant interest rate $r>0$.\\
Let $(S_t)_{t \geq 0}$ be the bitcoin-US exchange rate process also defined on the same filtered probability space and $(Y_t)_{t \geq 0}$  its associate log-prices process. They are related by:
\begin{equation}\label{eq:priceslog}
    S_t= S_0 \\exp(Y_t)
\end{equation}
For the latter we assume a mean-reverting dynamic under  the historic measure $P$ given by:
\begin{eqnarray} \label{eq:bsjmultid4}
    dY_t&=&\alpha(\mu-Y_t)dt+  dV_t
\end{eqnarray}
where $(V_t)_{t \geq 0}$ is a Levy process, to be specified latter on, $\mu$ and $\alpha$ are the mean-reverting level and rate respectively.\\
The following propositions provide well-known results about the characteristic function of the log-prices under the historic probability and the EMM defined via an Esscher transform. See for example Eberlein and Raible(1999) and Gerber and Shiu(1994).\\
  In order to select the EMM for pricing purposes we take an Esscher transform of the historic measure $P$. See Gerber and Shiu(1994) for a rationale in terms of a utility-maximization criteria. \\
      For a stochastic process $(X_t)_{t \geq 0}$ we consider its Esscher transform:
      \begin{equation}\label{eq:esscher}
    \frac{d \mathcal{Q}^{\theta}_t}{d P_t}=\exp(\theta X_t-t l_X(\theta)),\; 0 \leq t \leq T,\; \theta \in \mathbb{R}
  \end{equation}
   where $P_t$ and $\mathcal{Q}^{\theta}_t$ are the respective restrictions of $P$ and $\mathcal{Q}^{\theta}$ to the $\sigma$-algebra $\mathcal{F}_t$.\\
   %For the seek of simplicity we denote $\varphi^{0}_{X_{j \Delta}}:=\varphi_{X_{j \Delta}}, K_j(u, 0):=K_j(u), C_j(u,0):=C_j(0), j=1,2,3.,  %l_V^{\theta}(u)=l_V(u)$, etc.
      \begin{proposition}
    Let $(S_t)_{t \geq 0}$ be the process defined by equations with (\ref{eq:priceslog}) and (\ref{eq:bsjmultid4}). Let the process $(V_t)_{t \geq 0}$ have characteristic function $\varphi_{V_t}(u)$, Laplace exponent $l_V(u)$ under the probability $P$. \\
   Define by $\varphi^{\theta}_{V_t}$ and $ l^{\theta}_V(u)$  respectively the characteristic function and  Laplace exponent  of the process under the probability $\mathcal{Q}^{\theta}$ obtained by an Esscher transformation as given in equation (\ref{eq:esscher}). Then, the discounted price process $(\tilde{S}_t)_{t \geq 0}$ is a $\mathcal{Q}^{\theta}$-martingale  if for any $T>0$ the parameter $\theta$ verifies:
   \begin{equation}\label{eq:gershui}
    \int_0^{T} l^{\theta}_{V}( e^{-\alpha (T-s)})\;ds= rT-\mu(1-e^{-\alpha T})
   \end{equation}
   Moreover:
   \begin{eqnarray}\nonumber
        \varphi^{\theta}_{Y_t}(u)&=& \exp(i\mu (1-e^{-\alpha t})u-t l_V(\theta)+ I_t(u,\theta))\\ \label{eq:chfunlogprices}
        &&
      \end{eqnarray}
where:
    \begin{eqnarray*}
     I_t(u,\theta)&=& \int_0^t l_V(\theta+iu e^{-\alpha(t-s)})\;ds
  \end{eqnarray*}
   \end{proposition}
   \begin{proof}
    By Ito lemma, the solution of equation (\ref{eq:bsjmultid4}) is:
\begin{eqnarray}\label{eq:timechansol}
  Y_t&=&\mu (1-e^{-\alpha t}) + W_t
\end{eqnarray}
where $W_t=\int_0^t e^{-\alpha(t-s)}dV_s$.\\
We recall the following result  about the functional  of a Levy process $(\xi_t)_{t \geq 0}$ and a measurable function $f$:
\begin{equation}\label{eq:funclevy}
  E ( exp( i \int_0^t f(s)\;ds) )=exp(\int_0^t l_{\xi}(i f(s))\;ds)
\end{equation}
 Applied to the process $(W_t)_{t \geq 0}$ its characteristic function under the probability $\mathcal{Q}^{\theta}$ becomes:
\begin{eqnarray}  \label{eq:applevyk}
    \varphi^{\theta}_{W_t}(u)&=&  \exp(\int_0^t l^{\theta}_V(i u e^{-\alpha(t-s)}) ds)
 \end{eqnarray}
   By equation (\ref{eq:esscher}) combined with equations (\ref{eq:timechansol}) and (\ref{eq:applevyk}) the discounted process $(\tilde{S}_t)_{t \geq 0}$ is a $\mathcal{Q}^{\theta}$-martingale if and only if for any $0 \leq u < t $:
   \begin{eqnarray*}
   % \nonumber to remove numbering (before each equation)
    E_{\mathcal{Q}^{\theta}}(e^{W_t}/ \mathcal{F}_u)  &=& \exp(\mu(e^{-\alpha t}-e^{-\alpha u})+r(t-u))e^{W_u}\\
    \Leftrightarrow && \varphi^{\theta}_{W_{t-s}}(-i)=\exp(\mu(e^{-\alpha t}-e^{-\alpha u})+r(t-u))\\
   \Leftrightarrow && \int_0^{t-u} l^{\theta}_{V}( e^{-\alpha (t-s)})\;ds=\mu (e^{-\alpha t}-e^{-\alpha u})+r(t-u)
   \end{eqnarray*}
   In particular for $t=T$ and $u=0$ we have the result in equation (\ref{eq:gershui}).\\
   For the second part of the proposition we  simplify the notations and write $\mathcal{Q}^{\theta}:= \mathcal{Q}$.\\
   Next, notice that:
   \begin{eqnarray*}
    \varphi^{\theta}_{V_t}(u) &=&  E(e^{i u V_t} e^{\theta V_t-t l_V(\theta)})=\frac{\varphi_{V_t}(u- i \theta)}{\varphi_{V_t}(- i \theta)}\\
   \end{eqnarray*}
   and $ l^{\theta}_{V}(u) = l_V(u+\theta)-l_V(\theta)$.\\
  From equations (\ref{eq:timechansol}) and (\ref{eq:applevyk}):
    \begin{eqnarray*}
    \varphi^{\theta}_{Y_t}(u) &=& \exp \left[i\mu (1-e^{-\alpha t})u-t l_V(\theta)\right]\exp \left(\int_0^t l_V(\theta+ iu e^{-\alpha(t-s)})\;ds \right)
     \end{eqnarray*}
  \end{proof}
      \begin{remark}
   Notice that the characteristic function under the probability $P$ is obtained from equation (\ref{eq:chfunlogprices}) taking $\theta=0$. To simplify we write $I_t(u)=I_t(u,0)$, $\varphi^{0}_{Y_t}=\varphi_{Y_t}$ and $\mathcal{Q}^{0}=P$, etc.
   \end{remark}
     Parametric estimation is based on the log-return series given by:
 \begin{equation}\label{eq:logretpri}
   X_{j \Delta}=\log \left( \frac{S_{(j+1) \Delta}}{S_{j \Delta}} \right)=Y_{(j+1) \Delta}-Y_{j \Delta}, \; j=1,2,\ldots,n
 \end{equation}
 where $\Delta>0$ is the frequency at which the data is registered, typically daily observations. Notice that, because of the mean-reverting property, the observations are independent but not equally distributed.\\
The characteristic function of the log-returns is obtained in the following proposition:
 \begin{proposition}\label{prop:logretchf}
Let the log-returns series defined by equation (\ref{eq:logretpri}). For a model following equations (\ref{eq:priceslog}) and (\ref{eq:bsjmultid4}) and under the Esscher transformation the characteristic function of j-th log-return $ X_{j \Delta}$ is:
\begin{eqnarray}\label{eq:charflogret}
 \varphi^{\theta}_{ X_{j \Delta}}(u) &=&  \exp(C_1(u)+C_2(u,\theta)+C_3(u,\theta))
\end{eqnarray}
where:
\begin{eqnarray*}
% \nonumber to remove numbering (before each equation)
  C_1(u) &=& i u \mu e^{-\alpha j \Delta}(1-e^{-\alpha  \Delta})\\
  C_2(u,\theta) &=& \int_{j \Delta}^{(j+1) \Delta} l^{\theta}_V(i u e^{-\alpha ((j+1) \Delta-s)})\;ds \\
  C_3(u,\theta) &=& \int_0^{j \Delta} l^{\theta}_V(i u (e^{-\alpha \Delta}-1 )e^{-\alpha( j \Delta- s)})\;ds
\end{eqnarray*}
\end{proposition}
\begin{proof}
From equation (\ref{eq:timechansol}) we have:
\begin{eqnarray}\nonumber
  X_{j \Delta}&=& \mu e^{-\alpha j \Delta}(1-e^{-\alpha  \Delta}) + e^{-\alpha (j+1) \Delta}\int_0^{(j+1) \Delta} e^{\alpha s}\;dV_s- e^{-\alpha j \Delta} \int_0^{j \Delta} e^{\alpha s}\;dV_s\\ \nonumber
   &=&\mu e^{-\alpha j \Delta}(1-e^{-\alpha  \Delta}) + e^{-\alpha (j+1) \Delta}\int_{j \Delta}^{(j+1) \Delta} e^{\alpha s}\;dV_s+ e^{-\alpha j \Delta}(e^{-\alpha \Delta}-1 ) \int_0^{j \Delta} e^{\alpha s}\;dV_s\\ \label{eq:mrevlogret}
   &&
\end{eqnarray}
Hence, noting that $(W_t)_{t \geq 0}$ has independent increments:
\begin{eqnarray*}
 \varphi^{\theta}_{ X_{j \Delta}}(u) &=& E_{\mathcal{Q}} \left[\exp \left(i u(\mu e^{-\alpha j \Delta}(1-e^{-\alpha  \Delta})+e^{-\alpha (j+1) \Delta}\int_{j \Delta}^{(j+1) \Delta} e^{\alpha s}\;dV_s)\right) \right.\\
&& \left. \exp \left(- e^{-\alpha j \Delta}(1-e^{-\alpha \Delta}) \int_0^{j \Delta} e^{\alpha s}\;dV_s \right) \right]\\
&=& \exp[i u(\mu e^{-\alpha j \Delta}(1-e^{-\alpha  \Delta}))]E_{\mathcal{Q}} [\exp(i u e^{-\alpha (j+1) \Delta}\int_{j \Delta}^{(j+1) \Delta} e^{\alpha s}\;dV_s)]\\
 &&E_{\mathcal{Q}}\left[ \exp(-iu e^{-\alpha j \Delta}(1-e^{-\alpha \Delta}) \int_0^{j \Delta} e^{\alpha s}\;dV_s)\right]
 \end{eqnarray*}
  The conclusion follows  from equation (\ref{eq:applevyk}).
\end{proof}
\section{A jump-diffusion model for bitcoin-US dollar exchange}
We consider  a  jump-diffusion dynamics  for the Levy noise $(V_t)_{t \geq 0}$ given by:
\begin{equation}\label{eq:bsjmultid3}
    V_t=\sigma  B_t+Z_t
\end{equation}
where $(B_t)_{t \geq 0}$ is a Brownian motion and the process $(Z_t)_{t \geq 0}$ is a homogeneous compound Poisson process, independent of $(B_t)_{t \geq 0}$, such that:
\begin{equation}\label{eq:jumeq}
    Z_t=\sum_{k=1}^{N_t} \xi_k
\end{equation}
The process $(N_t)_{t \geq 0}$ is a Poisson process with  intensity $\lambda>0$, while $(\xi_k)_{k \in \mathbb{N}}$ is a sequence of i.i.d. random variables with common characteristic function $\varphi_X$. Furthermore we assume the existence of the moments up to order $M$ of the jumps, i.e. $E(\xi^k_1) < +\infty, k=1,2, \ldots, M$ and $\varphi_{\xi_1} \in L^1(\mathbb{R})$.
\begin{remark}
This model includes  the case of a single homogeneous compound Poisson with Gaussian jumps, leading to the classical Merton's model, see Merton(1976), or double exponential jump sizes, see Kou(2002).
\end{remark}
Results in section 2 are easily adapted to this setting. Notice that for the model described by equations (\ref{eq:bsjmultid3}) and (\ref{eq:jumeq}):
\begin{eqnarray*}
l_V(\theta+iu e^{-\alpha(t-s)}) &=& (\frac{1}{2} \sigma ^2 \theta^2-\lambda)+i \sigma ^2 \theta {\mathrm{e}}^{-\alpha \,\left(t-s\right)}u-\frac{1}{2}\sigma ^2 {\mathrm{e}}^{-2 \alpha(t-s)}u^2\\
&+& \lambda \varphi_{\xi}(- i \theta+u e^{-\alpha(t-s)}) \\
I_t(u,\theta) &=& \left(\frac{1}{2}\sigma ^2\,{\mathrm{\theta}}^2-\lambda \right)t+ i \frac{\sigma ^2 \theta}{\alpha }(1-e^{-\alpha t})u\\
     &-& \frac{1}{4 \alpha }\sigma ^2 (1-e^{-2\alpha t})u^2+ \lambda  \int_0^t  \,\mathrm{\varphi_{\xi}}(\mathrm{-i \theta}+ u\,{\mathrm{e}}^{-\alpha \,(t-s)})\;ds
  \end{eqnarray*}
  Therefore:
\begin{eqnarray}\nonumber
 \varphi^{\theta}_{Y_t}(u) &=& \exp \left[-\lambda \varphi_{\xi}(\theta) t+ i \left(\frac{\mu+\sigma ^2 \theta}{\alpha } \right)(1-e^{-\alpha t})u \right. \\ \nonumber
 &-& \left. \frac{1}{4 \alpha }\sigma ^2 (1-e^{-2\alpha t})u^2+ \lambda  \int_0^t  \,\varphi_{\xi}(- i \theta+u e^{-\alpha(t-s)})\;ds \right] \\ \label{eq:charfunctjdiffesscher}
 &&
\end{eqnarray}
Moreover:
\begin{eqnarray*}
% \nonumber to remove numbering (before each equation)
 \int_0^T l^{\theta}_V( e^{-\alpha(T-s)})\;ds &=&  \frac{\sigma ^2}{4 \alpha }(1-e^{-2\alpha T})+\frac{\sigma ^2 \theta}{\alpha }(1-e^{-\alpha T})\\
 &+& \lambda \int_0^T  \,\mathrm{\varphi_{\xi}}(\mathrm{\theta}+{\mathrm{e}}^{-\alpha \,(T-s)})\;ds-\lambda  \varphi_{\xi}(\theta)T
\end{eqnarray*}
From proposition 1, equation (\ref{eq:gershui}), $\theta$ solves the equation:
\begin{eqnarray*}
 &&\lambda \int_0^T \mathrm{\varphi_{\xi}}\left(\mathrm{\theta}+ {\mathrm{e}}^{-\alpha \,\left(T-s\right)}\right) \,d s = (\lambda \varphi_{\xi}(\theta)+r)T\\
 &-& (\mu+\frac{\sigma ^2 \theta}{\alpha })(1-e^{-\alpha T})-\frac{\sigma ^2}{4 \alpha }(1-e^{-2\alpha T})
 \end{eqnarray*}
Next, we compute the intermediate quantities:
 \begin{eqnarray*}
  C_2(u, \theta)&=&\int_{j \Delta}^{(j+1) \Delta} l_V(\theta+i u e^{-\alpha ((j+1) \Delta-s)})\;ds-l_V(\theta)\Delta\\
  &=& ( \frac{1}{2} \sigma ^2 \theta^2-\lambda-l_V(\theta))\Delta +i \frac{\sigma ^2 \theta}{\alpha} (1-{\mathrm{e}}^{-\alpha \Delta})u\\
&-&\frac{\sigma ^2}{4 \alpha} (1-{\mathrm{e}}^{-2 \alpha \Delta})u^2+\lambda  \int_{j \Delta}^{(j+1) \Delta} \varphi_{\xi}(-i\theta+u e^{-\alpha((j+1)\Delta-s)})\;ds \\
  \end{eqnarray*}
 \begin{eqnarray*}
  C_3(u, \theta)&=& \int_{0}^{j \Delta} l_V(\theta+i u e^{-\alpha ((j+1) \Delta-s)})\;ds-l_V(\theta)j \Delta\\
   &=& (\frac{1}{2} \sigma ^2 \theta^2-\lambda-l_V(\theta))j \Delta-i \frac{\sigma ^2 \theta}{\alpha} (1-{\mathrm{e}}^{-\alpha j\Delta})(1-e^{-\alpha \Delta})u\\
&-&\frac{\sigma ^2 }{4 \alpha}(1-{\mathrm{e}}^{-2 \alpha j \Delta })(1-e^{-\alpha \Delta})^2u^2+\lambda \int_0^{j \Delta}  \varphi_{\xi}(-i\theta+u (e^{-\alpha \Delta}-1)e^{-\alpha(j \Delta-s)})\;ds
\end{eqnarray*}
 Hence, from proposition \ref{prop:logretchf}  we have:
\begin{eqnarray}\nonumber
   \varphi^{\theta}_{X_{j \Delta}}(u)  &=&  \exp \left[ \lambda  (K_1(u,\theta)+K_2(u,\theta)) \right.\\ \nonumber
  &-& \left.\lambda \varphi_{\xi}(\theta)(j+1)\Delta +i (\mu+\frac{\sigma ^2 \theta}{\alpha}) e^{-\alpha j \Delta}(1-{\mathrm{e}}^{-\alpha \Delta})u-\frac{\sigma ^2}{4 \alpha} E_{j,2}(\alpha)u^2 \right] \\ \label{eq:chfjumpdiff2}
    &&
 \end{eqnarray}
where:
\begin{eqnarray*}
% \nonumber to remove numbering (before each equation)
  K_1(u,\theta) &=&  \int_{j \Delta}^{(j+1) \Delta} \varphi_{\xi}(-i\theta+u e^{-\alpha((j+1)\Delta-s)})\;ds \\
   K_2(u,\theta) &=& \int_0^{j \Delta}  \varphi_{\xi}(-i\theta+u (e^{-\alpha \Delta}-1)e^{-\alpha(j \Delta-s)})\;ds\\
    E_{j,k}(\alpha)&=&(1-e^{-k \alpha \Delta})+(-1)^k (1-e^{-\alpha \Delta})^k(1-e^{-k \alpha j \Delta} )
\end{eqnarray*}
In particular for $\theta=0$:
\begin{eqnarray}\nonumber
   \varphi_{X_{j \Delta}}(u)  &=&  \exp \left[ \lambda  (K_1(u)+K_2(u))-\lambda (j+1)\Delta \right.\\ \nonumber
  &+&\left.  i \mu e^{-\alpha j \Delta}(1-{\mathrm{e}}^{-\alpha \Delta})u-\frac{\sigma ^2}{4 \alpha} E_{j,2}(\alpha)u^2 \right] \\ \label{eq:chfjumpdiffzero}
    &&
 \end{eqnarray}
\begin{example}\textit{Mean-reverting Black-Scholes  model}\\
Although it is clear from the empirical analysis in section 4 below that a mean-reverting Black-Scholes model does not capture the dynamic of bitcoin-US dollar exchange rate, nonetheless we consider the latter for comparison.\\
 To this end we set $V_t=\sigma B_t$. Hence:
\begin{eqnarray*}
\varphi^{\theta}_{Y_t}(u) &=& \exp \left( i \left(\frac{\sigma ^2 \theta}{\alpha }+\mu \right)(1-e^{-\alpha t})u- \frac{1}{4 \alpha }\sigma ^2 (1-e^{-2\alpha t})u^2 \right)
\end{eqnarray*}
Therefore, the Gerber-Shui parameter $\theta$ solves:
\begin{eqnarray*}
 r T-(\frac{\sigma ^2 \theta}{\alpha }+\mu)(1-e^{-\alpha T})-\frac{\sigma ^2}{4 \alpha }(1-e^{-2\alpha T})&=&0
 \end{eqnarray*}
 Hence:
\begin{eqnarray*}
% \nonumber to remove numbering (before each equation)
  \theta &=& \frac{\alpha}{\sigma^2} \left( rT (1-e^{-\alpha T})^{-1}-\frac{\sigma ^2}{4 \alpha }(1+e^{-\alpha T})-\mu\right)
\end{eqnarray*}
and
 \begin{eqnarray}\nonumber
   \varphi^{\theta}_{X_{j \Delta}}(u)  &=&  \exp \left[i (\mu+\frac{\sigma ^2 \theta}{\alpha}) e^{-\alpha j \Delta}(1-{\mathrm{e}}^{-\alpha \Delta})u- \frac{\sigma ^2}{4 \alpha} E_{j,2}(\alpha) u^2 \right]
  \end{eqnarray}
\end{example}
  \begin{example}\textit{Mean-reverting jump-diffusion model with Gaussian jumps.}\\
We assume  $\xi_k \sim N(\mu_J, \sigma^2_J)$. Then:
\begin{eqnarray}\nonumber
\varphi_{\xi}(u)&=& \exp(i \mu_J u-\frac{1}{2}\sigma^2_J u^2)\\ \nonumber
\int_0^t \varphi_{\xi}(-i\theta+u e^{-\alpha(t-s)})\;ds &=& \varphi_{\xi}(-i \theta)\int_0^t \varphi_{\xi}(u e^{-\alpha(t-s)}) \exp(-i\sigma^2_J \theta u e^{-\alpha(t-s)})\;ds\\  \nonumber
  &=& \frac{\varphi_{\xi}(-i \theta)}{\alpha} A_1(u,-i \sigma^2_J \theta u,t)\\ \label{eq:intcharfun}
  &&
  \end{eqnarray}
  after the change of variable  $y= e^{-\alpha(t-s)}$, where:
  \begin{eqnarray*}
   A_1(u,v,t)=\int_{e^{-\alpha t}}^1 y^{-1} \varphi_{\xi}(u y) \exp(-v y)\;ds
  \end{eqnarray*}
  Therefore, combining equations (\ref{eq:charfunctjdiffesscher}) and (\ref{eq:intcharfun}):
   \begin{eqnarray*}
   \varphi^{\theta}_{Y_t}(u) &=& \exp \left[-\lambda \varphi_{\xi}(\theta)t+ i \left(\frac{\sigma ^2 \theta}{\alpha }+\mu \right)(1-e^{-\alpha t})u - \frac{1}{4 \alpha }\sigma ^2 (1-e^{-2\alpha t})u^2 \right. \\ \nonumber
  &+&  \left. \frac{\lambda \varphi_{X}(-i \theta)}{\alpha} A_1(u,i \sigma^2_J \theta u,t) \right]
    \end{eqnarray*}
  Similar calculations lead to:
   \begin{eqnarray*}
   % \nonumber to remove numbering (before each equation)
     \int_0^T \varphi_{\xi}(\theta+u e^{-\alpha(T-s)})\;ds &=&  \frac{\varphi_{\xi}(\theta)}{\alpha}A_1(1,\sigma^2_J \theta u,T)
   \end{eqnarray*}
        The Gerber-Shui coefficient $\theta_{GS}$ satisfies:
   \begin{eqnarray*}
     \lambda \varphi_{\xi}(\theta) (\alpha T-A_1(1,\sigma^2_J \theta,T))-(\sigma^2 \theta+\alpha \mu)(1-e^{-\alpha T}) &=& \frac{\sigma^2}{4}(1-e^{-2\alpha T})-\alpha r T
   \end{eqnarray*}
   Finally, the characteristic function under the probability $\mathcal{Q}^{\theta}$ of the log-returns is written:
   \begin{eqnarray*}\nonumber
   \varphi^{\theta}_{X_{j \Delta}}(u)  &=&  exp \left[\lambda \frac{\varphi_{\xi}(-i \theta)}{\alpha} \left[ A_1(u,i \sigma^2_J \theta u,(j+1)\Delta) -  A_1(u,i\sigma^2_J \theta u,j\Delta) \right] \right.\\ \nonumber
   &+& \lambda \frac{\varphi_{\xi}(-i \theta)}{\alpha}  A_1(u(e^{-\alpha \Delta}-1),i\sigma^2_J \theta(e^{-\alpha \Delta}-1) u,j \Delta)\\ \nonumber
 &-& \left. \lambda \varphi_{\xi}(-i \theta)(j+1)\Delta +i (\mu+\frac{\sigma ^2 \theta}{\alpha}) e^{-\alpha j \Delta}(1-{\mathrm{e}}^{-\alpha \Delta})u - \frac{\sigma ^2}{4 \alpha} E_{j,2}(\alpha) u^2 \right]
 \end{eqnarray*}
  \end{example}
  \begin{example}\textit{Mean-reverting jump-diffusion model with double exponential jumps.}\\
In the case of the Kuo model, the common p.d.f. of the jump sizes is described by:
\begin{eqnarray*}\nonumber
f_{X}(x) &=& q \eta_1 e^{-\eta_1 x}1_{\{x  \geq 0\}} + (1-q)\eta_2 e^{\eta_2 x}1_{\{x < 0\}} \hspace{7mm} \eta_1 >1,\eta_2 >0 \\ \label{eq:marginkou1}
&&
\end{eqnarray*}
where $q$ and $1-q$ represent the respective probabilities of upward and downward jumps.
The characteristic function of the jumps is:
\begin{eqnarray*}
  \varphi_{X}(u)  &=& \lambda \left(\frac{q \eta_1}{\eta_1-iu}+\frac{(1-q) \eta_2}{\eta_2+iu} \right)
\end{eqnarray*}
Hence:
\begin{eqnarray*}\nonumber
&& \int_0^t \varphi_{\xi}(-i \theta+u e^{-\alpha(t-s)})\;ds = q \eta_1 \int_0^t \frac{1}{\eta_1+u e^{-\alpha(t-s)}-i \theta}\;ds\\
&+& (1-q) \eta_2 \int_0^t \frac{ 1}{\eta_2+u e^{-\alpha(t-s)}+ \theta}\;ds\\
&=& \frac{q \eta_1}{\alpha} \int_{e^{-\alpha t}}^1 \frac{1}{y(\eta_1-iuy - \theta)}\;dy + \frac{(1-q) \eta_2}{\alpha} \int_{e^{-\alpha t}}^1 \frac{1}{y(\eta_2+ i u y- \theta)}\;dy\\
&=& \frac{q \eta_1}{\alpha}A_2(u,\theta,t)+ \frac{(1-q) \eta_2}{\alpha}A_3(u,\theta,t)
  \end{eqnarray*}
  where:
  \begin{eqnarray*}
      A_2(u,\theta,t)  &=& \int_{e^{-\alpha t}}^1 \frac{1}{y(\eta_1-i uy - \theta)}\;dy \\
    A_3(u,\theta,t)  &=& \int_{e^{-\alpha t}}^1 \frac{1}{y(\eta_2+ i u y+ \theta)}\;dy
  \end{eqnarray*}
  Then:
  \begin{eqnarray*}
   \varphi^{\theta}_{Y_t}(u) &=& exp \left[-\lambda \varphi_{\xi}(\theta) t+ i \left(\frac{\sigma ^2 \theta}{\alpha }+\mu \right)(1-e^{-\alpha t})u \right. \\ \nonumber
 &-& \left. \frac{1}{4 \alpha }\sigma ^2 (1-e^{-2\alpha t})u^2+ \lambda  \frac{q \eta_1}{\alpha}A_2(u,\theta,t)+  \lambda \frac{(1-q) \eta_2}{\alpha}A_3(u,\theta,t)\right] \\
    \end{eqnarray*}
  Similar calculations lead to:
   \begin{eqnarray*}
  \int_0^T \varphi_{\xi}( e^{-\alpha(T-s)})\;ds &=& \frac{q \eta_1}{\alpha} A_2(-i,\theta,T)+\frac{(1-q) \eta_2}{\alpha}A_3(i, \theta, T)
   \end{eqnarray*}
   Therefore, the Gerber-Shui coefficient $\theta_{GS}$ verifies:
    \begin{eqnarray*}
&& \frac{q \eta_1}{\alpha} A_4(-i,\theta,T)+\frac{(1-q) \eta_2}{\alpha}A_5(i, \theta, T)\\
  &=&(\lambda \varphi_{\xi}(\theta)+r)T-(\frac{\sigma ^2 \theta}{\alpha }+\mu)(1-e^{-\alpha T})-\frac{\sigma ^2}{4 \alpha }(1-e^{-2\alpha T})
 \end{eqnarray*}
 Finally:
 \begin{eqnarray}\nonumber
   \varphi^{\theta}_{X_{j \Delta}}(u)  &=&  exp \left[\frac{q \eta_1}{\alpha}[A_2(u,\theta,(j+1)\Delta)-A_2(u,\theta,j \Delta)]\right.\\ \nonumber
   &+& \frac{(1-q) \eta_2}{\alpha}[A_3(u,\theta,(j+1)\Delta)-A_3(u,\theta,j \Delta)]\\ \nonumber
   &+&  \frac{q \eta_1}{\alpha}A_2(u(e^{-\alpha \Delta}-1),\theta,j\Delta)+\frac{(1-q) \eta_2}{\alpha}A_3(u(e^{-\alpha \Delta}-1),\theta,j\Delta)\\ \nonumber
 &-& \left. \lambda \varphi_{\xi}(\theta)(j+1)\Delta +i (\mu+\frac{\sigma ^2 \theta}{\alpha}) e^{-\alpha j \Delta}(1-{\mathrm{e}}^{-\alpha \Delta})u
- \frac{\sigma ^2}{4 \alpha} E_{j,2}(\alpha) u^2 \right]  \\ \label{eq:chfjumpdiffkou}
    &&
 \end{eqnarray}
\end{example}
  \section{Parameter calibration}
 We present a brief  statistical analysis and estimate the parameters in the model given by equation (\ref{eq:bsjmultid4}). The analysis is based on a historic series of daily bitcoin-US dollar exchange rate  from quotations in BTC, expanding from January 2011 to June 2018. The data is taken from  www.candainvestment.com web site. In figure \ref{fig:tsexcahnges} we can see the corresponding exchange(left) and log-return(right)  series. Trade volume, volatility and large oscillations have dramatically increased in recent years, specially after 2017. \\
 We look at some empirical features. Daily closure exchange rates and log-return exchange rates first four moments are shown in table \ref{tab:descritive stat}. It reveals an asymmetric  probability distribution, skewed to the right,  with a remarkable high kurtosis.\\

\begin{table}[hb!]
  \centering
 \begin{tabular}{|c|c|c|c|c|}
  \hline
  % after \\: \hline or \cline{col1-col2} \cline{col3-col4} ...
 series  & mean & volatility& skewness &kurtosis\\ \hline
Bitcoin-US exchange &     1444  &  2873.2 &   2.9  &  11.2 \\  \hline
log-returns & 0.0031 &   0.0752  &  3.0083  & 125.7292 \\ \hline
\end{tabular}
\caption{ Average, volatility, skewness and kurtosis of bitcoin US dollar exchange and the log-returns of the prices, January 2011-June 2018}\label{tab:descritive stat}
\end{table}

\begin{figure}
\begin{center}
\subfigure[]{
\resizebox*{6cm}{!}{\includegraphics{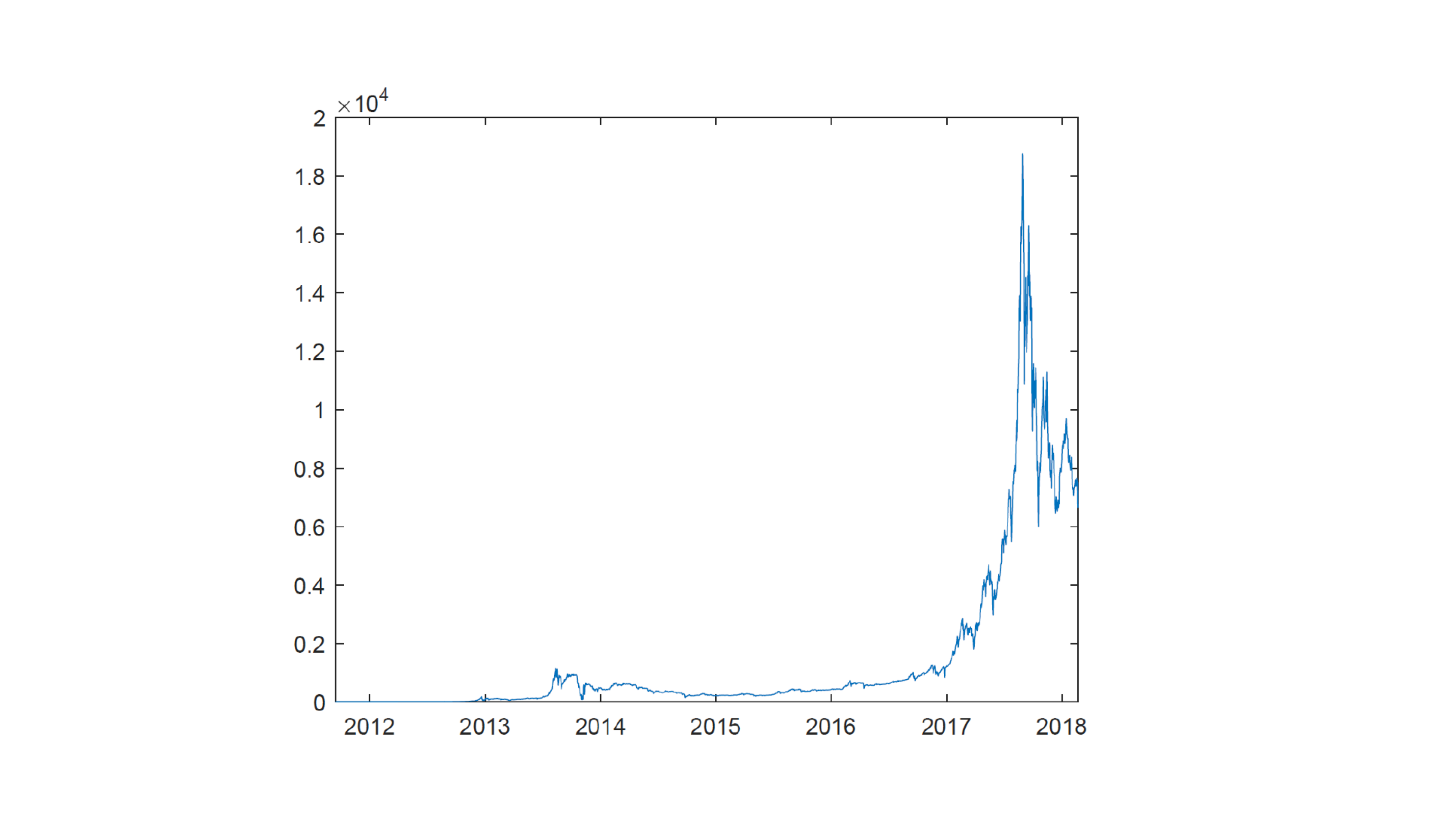}}}\hspace{5pt}
\subfigure[]{
\resizebox*{6cm}{!}{\includegraphics{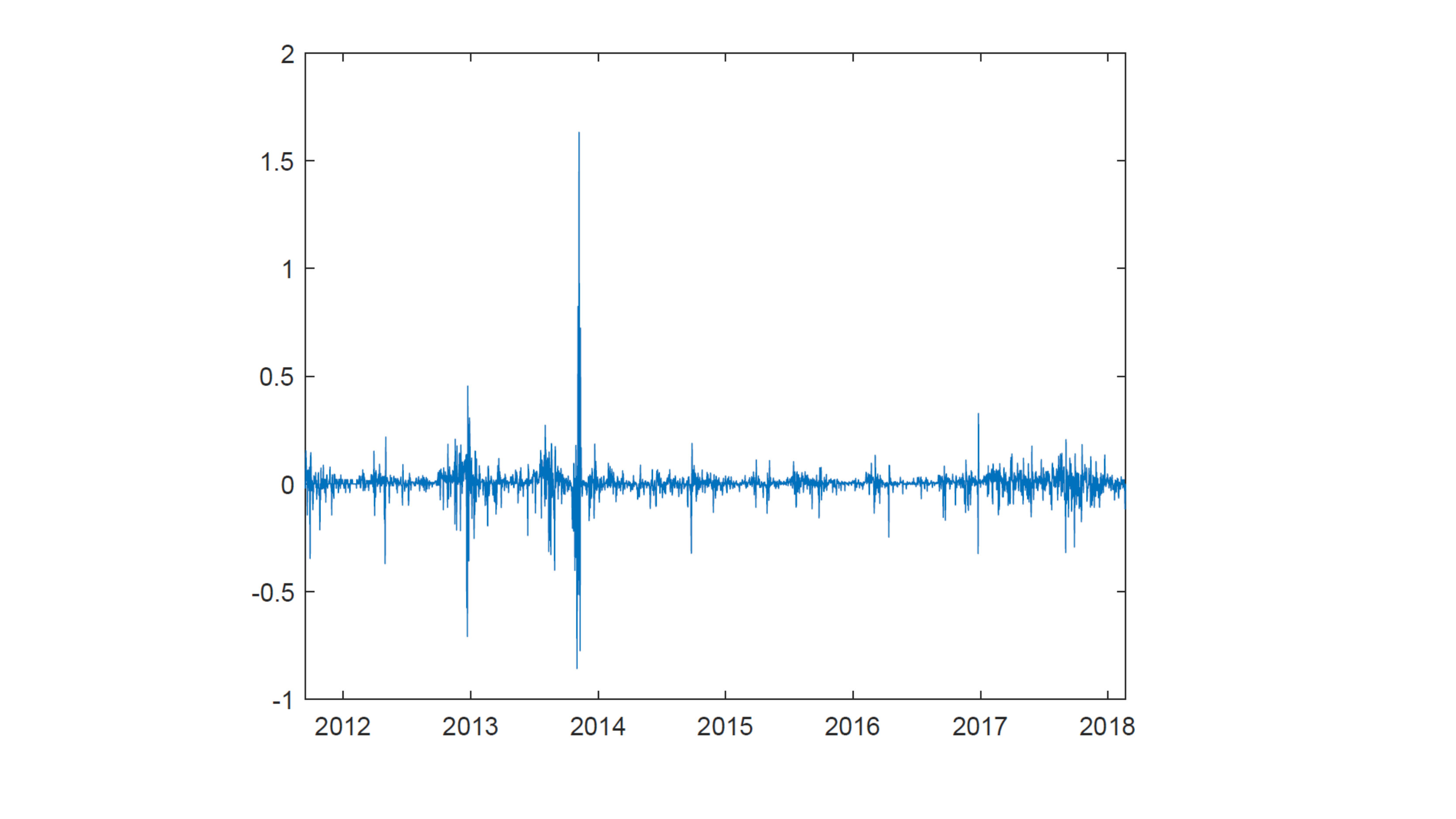}}}\hspace{5pt}
\caption{ }
\label{fig:tsexcahnges}
\end{center}
\end{figure}

In figure \ref{fig:autocorr} the autocorrelation series of log-returns(left) is shown. Most values lie within the zero confidence strip at 95\%. As it is common in most financial series the autocorrelation of the  squared log-returns(right) is significant for most relevant lags. It provides an argument of non-Gaussianity that is confirmed by  a Kolmogorov-Smirnov test for log-returns. It rejects the hypothesis of normality with a $p-value =4.2987 \times 10^{-50}$, statistics $k =0.1577$ and a critical value $c = 0.0256$.\\

\begin{figure}
\begin{center}
\subfigure[]{
\resizebox*{6cm}{!}{\includegraphics{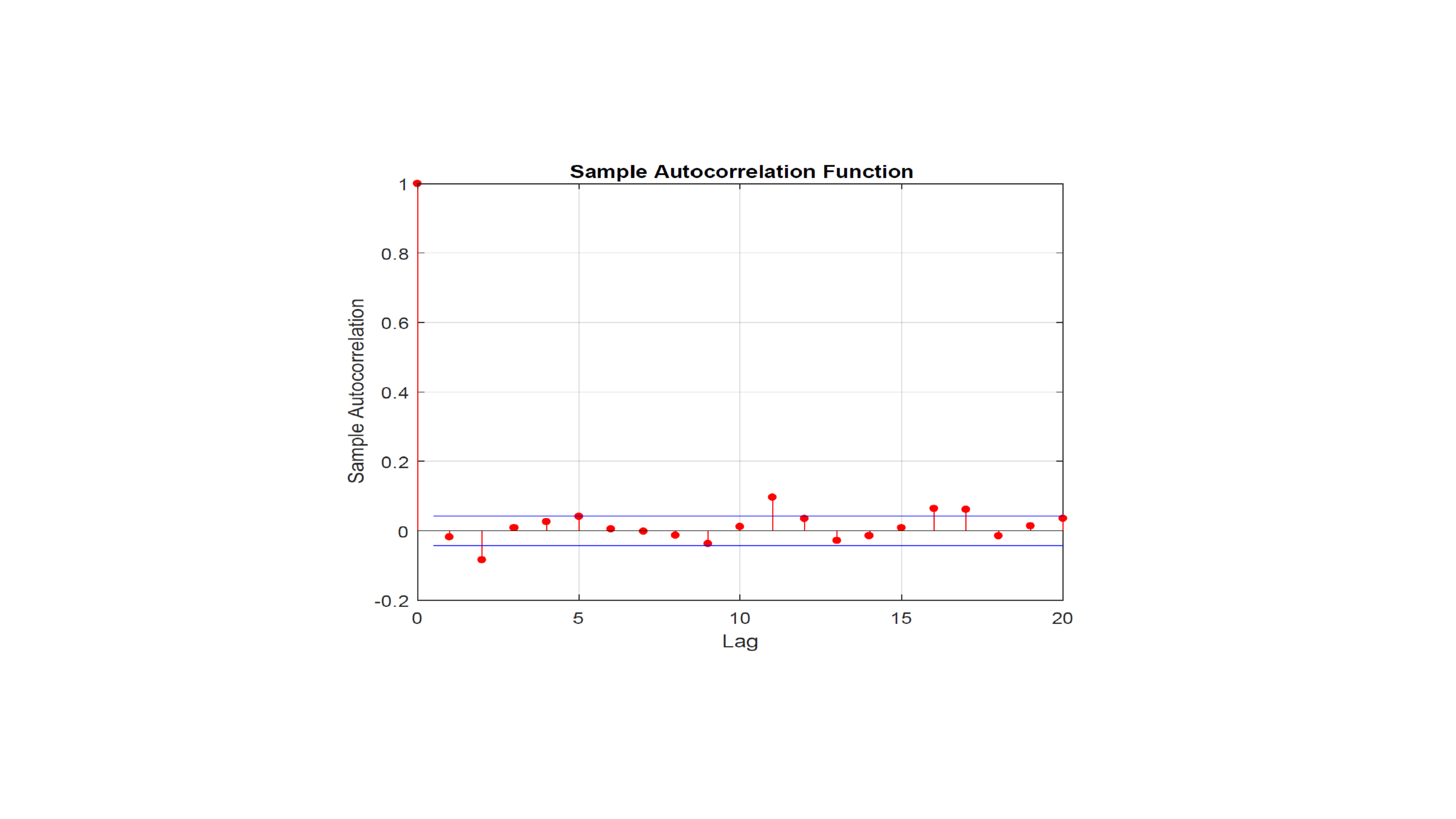}}}\hspace{5pt}
\subfigure[]{
\resizebox*{6cm}{!}{\includegraphics{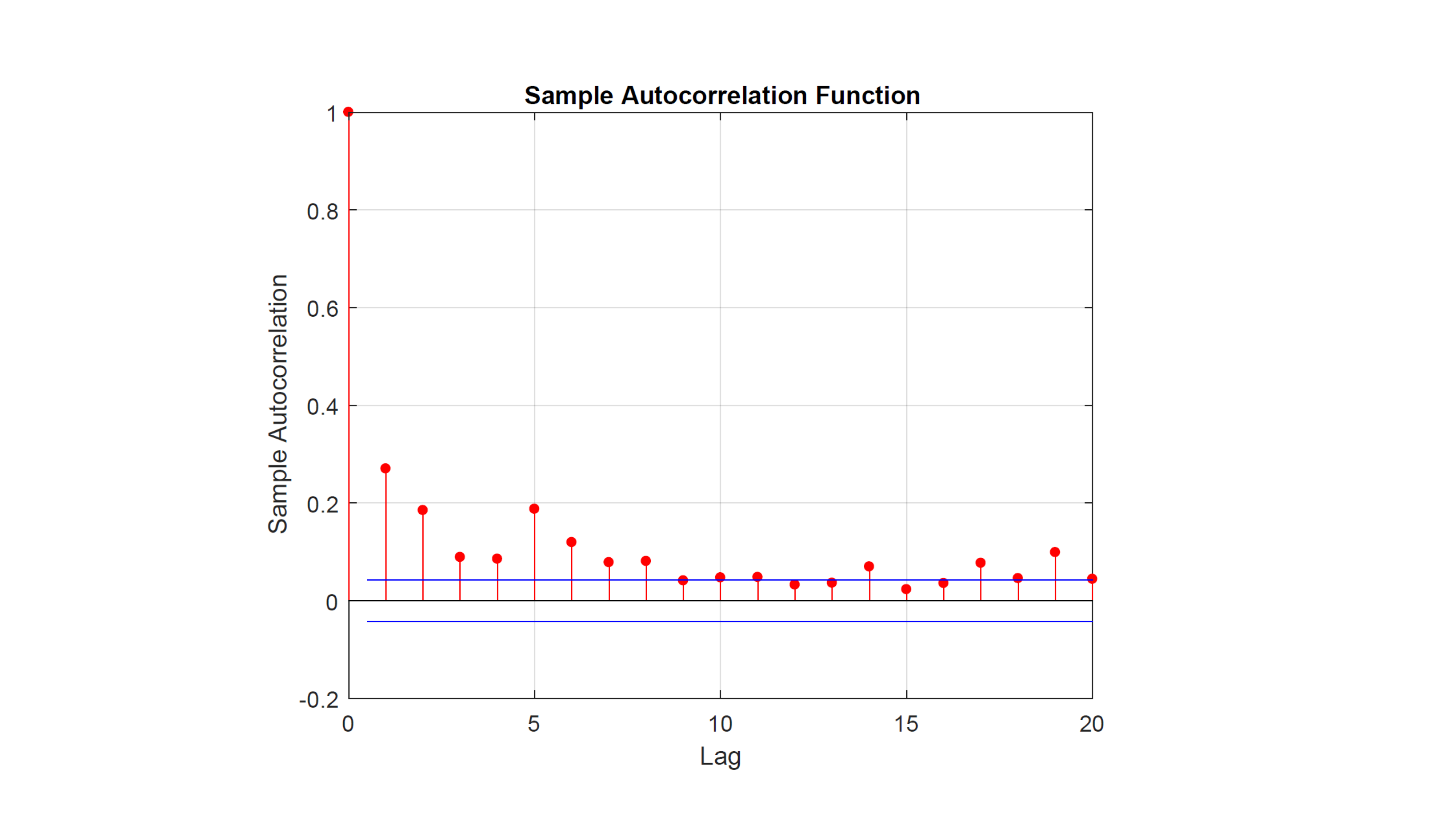}}}
\caption{Autocorrelation of log-returns and squared log-returns }
\label{fig:autocorr}
\end{center}
\end{figure}

In figure \ref{fig:pdfemp} the  graph in the left shows the empirical probability density function (p.d.f.) of log-return exchanges compared to a normal p.d.f. with the same mean and standard deviation. Again, it suggests a non-Gaussian distribution that allows to capture large oscillations and heavy tails present in the data. The  graph in the right shows a scaled and shifted t-student p.d.f., red line, and a stable p.d.f., blue line adjusted to the bitcoin data. Both probability distributions provide a better fit than the normal one. \\

 Parameters in the fitting of the log-returns p.d.f. are estimated using a maximum likelihood approach. Notice that in the case of the stable distribution the p.d.f. is not explicitly known. Numerical inversion of the characteristic function is required. See for example Mittnik and Rachev (2001).\\
The empirical p.d.f. of log-return exchanges is obtained using a non-parametric Gaussian kernel.

\begin{figure}
\begin{center}
\subfigure[]{
\resizebox*{6cm}{!}{\includegraphics{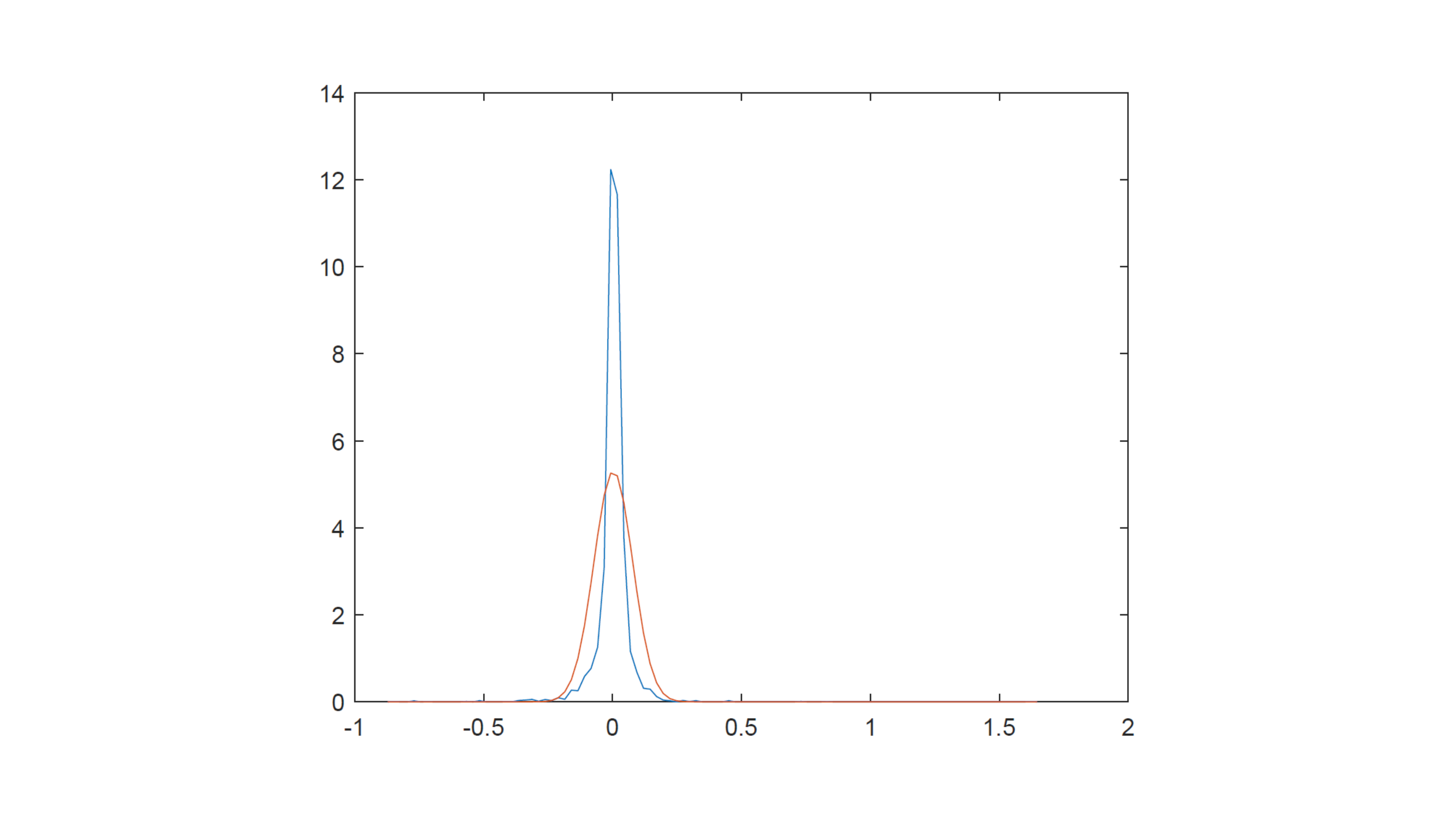}}}\hspace{5pt}
\subfigure[]{
\resizebox*{6cm}{!}{\includegraphics{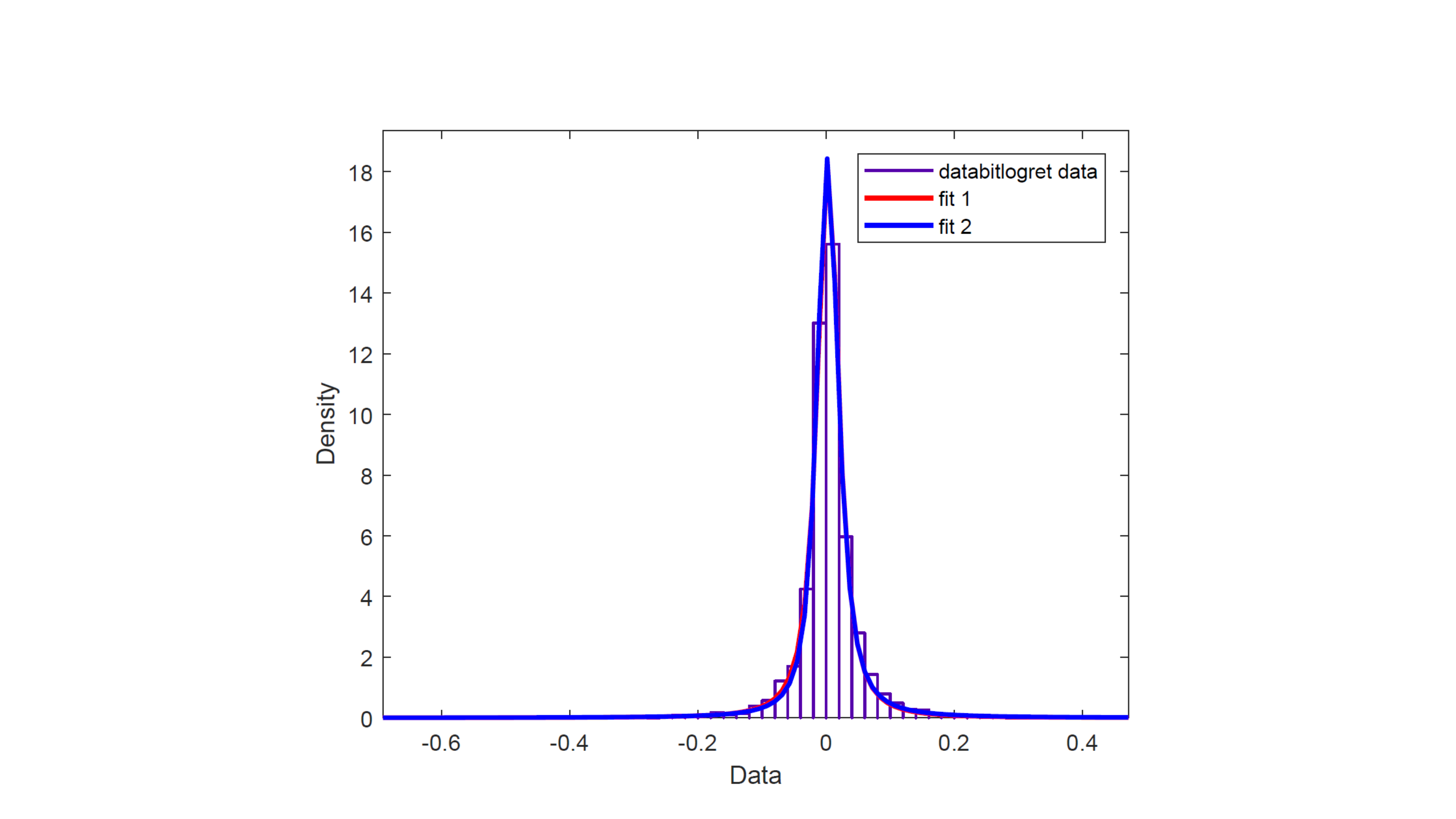}}}\hspace{5pt}
\caption{Left: empirical pdf of log-return exchange bitcoin-US dollar, compared with a normal pdf. Right: Empirical p.d.f. vs stable and t-student distributions   }
\label{fig:pdfemp}
\end{center}
\end{figure}

Parameter estimation results are shown in table \ref{tab:mlestudstab} for a t-student and a stable distribution. Between brackets the 95\% confidence interval of the estimation, accordingly to the Fisher information estimated from a maximum likelihood approach.  The values of parameters $\alpha$ in the stable distribution and the degrees of freedom  $\nu$  in a located and scaled normal distribution shows a strong heavy-tailed distribution of the exchanges.\\
In the case of a t-student distribution parameter $\alpha$ represents its number of degrees of freedom. For both, stable and t-student, a value of $\alpha$ that  low indicated an extreme high volatility and tail thickness.\\
The results above are confirmed by a fit based on a generalized Pareto distribution. In this case the parameter $\alpha$ means the shape of the distribution. A positive value $\alpha=$ indicates a heavy-tailed probability distribution. Data in excess of $0.05$ have been considered.
\small{
\begin{table}[hb!]
  \centering
  \begin{tabular}{|c|c|c|c|c|}
    \hline
    % after \\: \hline or \cline{col1-col2} \cline{col3-col4} ...
   param.  & $\nu$ & $\sigma$ & $\mu$ & \\ \hline
      t-student  & 1.35307   &  0.0174    &  0.0056     &  \\ \hline
     conf. int.  &[1.23695, 1.4801]  & [0.01626, 0.01867] & [0.00467, 0.00653] & \\ \hline
   param.  & $\alpha$ & $\beta$ & $\sigma$ & $\mu$ \\ \hline
    stable &  1.13346     & 0.00306    &  0.0157503    & 0.00564    \\ \hline
   conf. int. & [1.08219, 1.18473] & [-0.07689, 0.08300] & [0.01491, 0.0166]& [0.00465, 0.0066] \\    \hline
  \end{tabular}

  \caption{Maximum likelihood estimates of parameters in a scaled t-student and stable   laws}\label{tab:mlestudstab}
\end{table}
}
%Guillaume et al, Longin and Loretan and Phillips.(read the papers)\\
%Stable distribution
 %   alpha =
  %   beta = 0.00305589   [-0.0768911, 0.0830029]
   %   gam =  0.0157503   [0.0149057, 0.0165949]
   % delta = 0.00564373   [0.00464929, 0.00663817]
   % t Location-Scale distribution
    %   mu = 0.00559924   [0.00466492, 0.00653357]
    %sigma =  0.0174221   [0.0162598, 0.0186675]
     %  nu =    1.35307   [1.23695, 1.4801]
 The parameters  considered are listed in table \ref{tab:parm}. They correspond respectively to a mean-reverting Black-Scholes, Merton and Kou models.

\begin{table}[htp!]
  \centering
  \begin{tabular}{|c|c|}
    \hline
    % after \\: \hline or \cline{col1-col2} \cline{col3-col4} ...
    Model & Parameters \\     \hline
   BShMR  & $\mu, \alpha, \sigma$ \\     \hline
   MeMR  & $\mu, \alpha, \sigma, \mu_J,\sigma_J, \lambda$ \\     \hline
   KouMR  & $\mu, \alpha, \sigma, \eta_1, \eta_2, q$  \\   \hline
   \end{tabular}
  \caption{Parameters in different models}\label{tab:parm}
\end{table}

\subsection{Method of Moments}
We match empirical and theoretical  moments. Theoretical moments are obtained via the derivative of the characteristic function of the log-returns in equation (\ref{eq:chfjumpdiffzero}). Notice that we are estimating the parameters under the historic measure. Hence, we have:
\begin{eqnarray*}
% \nonumber to remove numbering (before each equation)
 D^k K_1(0)  &=&  \varphi^{(k)}_{\xi}(0)\int_{j \Delta}^{(j+1) \Delta} e^{-k \alpha((j+1)\Delta-s)}  \;ds\\
 &=&  \varphi^{(k)}_{\xi}(0) \frac{(1-e^{-k \alpha \Delta} )}{k \alpha }=i^k E(\xi^k)\frac{(1-e^{-k \alpha \Delta} )}{k \alpha }\\
 D^k K_2(0)  &=& \varphi^{(k)}_{\xi}(0) (e^{-\alpha \Delta}-1)^k \int_0^{j \Delta}  e^{-k \alpha(j \Delta-s)}\;ds\\
 &=& (-1)^k i^k E(\xi^k)  (1-e^{-\alpha \Delta})^k  \frac{(1-e^{-k \alpha j \Delta} )}{k \alpha }
\end{eqnarray*}
On the other hand, after defining:
\begin{eqnarray*}
% \nonumber to remove numbering (before each equation)
  T_1(u) &=&  -\frac{1}{2} \lambda (j+1)\Delta +i \mu e^{-\alpha j \Delta}(1-{\mathrm{e}}^{-\alpha \Delta})u\\   \nonumber
&-& \frac{\sigma ^2}{4 \alpha} E_{j,2}(\alpha)u^2 + \lambda (K_1(u)+K_2(u))
\end{eqnarray*}
we get:
\begin{eqnarray*}
% \nonumber to remove numbering (before each equation)
 D T_1(0) &=& i \mu e^{-j \alpha \Delta}(1-e^{-\alpha \Delta})-i \lambda   \frac{E(\xi)}{ \alpha} E_{j,1}(\alpha) \\
 D^2 T_1(0) &=&- \frac{1}{2 \alpha} \left(\lambda E(\xi^2) +\sigma ^2 \right)E_{j,2}(\alpha)\\
  D^k T_1(0) &=&i^k \lambda   \frac{E(\xi^k)}{k \alpha} E_{j,k}(\alpha), k=3,4,\ldots
\end{eqnarray*}
The  derivatives of the characteristic function can be computed recursively by:
\begin{eqnarray*}
% \nonumber to remove numbering (before each equation)
 D^k \varphi_{X_{j \Delta}}(u)  &=& \sum_{l=0}^{k-1} \binom{k-1}{l}D^{l+1} T_1(u)D^{k-l-1} \varphi_{X_{j \Delta}}(u)
\end{eqnarray*}
Details in the calculation of the first moments are presented in the appendix. \\
 Next, we define the  empirical moments with respect to the origin in a natural way and match to as many theoretical moments as needed. Hence:
\begin{eqnarray*}
  \hat{m}_{k} &=& \frac{1}{n} \sum_{j=1}^n X^k_{j \Delta}=\frac{1}{n}\sum_{j=1}^n E(X^k_{j \Delta}):=\mu_k \;, k \in \mathbb{N}
\end{eqnarray*}
It leads to the equations:
\begin{eqnarray*}
 \hat{m}_1&=& (\mu (1-e^{-\alpha  \Delta}) - \lambda  \frac{E(\xi)}{ \alpha}) \overline{e^{-\alpha j \Delta}}\\
 &=& (\mu- \lambda  \frac{E(\xi)}{ \alpha}) (1-e^{-\alpha  \Delta}) \overline{e^{-\alpha j \Delta}}\\
  \hat{m}_2 &=& -\frac{\overline{E_{j,2}(\alpha)}}{2 \alpha} \left(\lambda E(\xi^2) +\sigma ^2 \right) - \mu^2  (1-e^{-\alpha \Delta})^2 \overline{e^{-2 \alpha j \Delta}}\\
  &+& 2 \frac{\lambda \mu E(\xi)}{ \alpha} \overline{e^{- \alpha j \Delta}E_{j,1}(\alpha)}-\frac{\lambda^2  (E(\xi))^2}{ \alpha^2} \overline{E^2_{j,1}(\alpha)}
  \end{eqnarray*}

  \begin{eqnarray*}
 \hat{m}_3  &=&     \frac{ \lambda}{3 \alpha} E(\xi^3) \overline{E_{j,3}(\alpha)}+  \frac{3 \mu}{2 \alpha} \left(\lambda E(\xi^2) +\sigma ^2 \right) (1-e^{-\alpha \Delta}) \overline{e^{-j \alpha \Delta} E_{j,2}(\alpha)}\\
 &-& \frac{3 \lambda}{ 2 \alpha^2} E(\xi)\left(\lambda E(\xi^2) +\sigma ^2 \right) \overline{E_{j,2}(\alpha) E_{j,1}(\alpha)}+ \overline{ \left(   \mu e^{-j \alpha \Delta}(1-e^{-\alpha \Delta})-\lambda   \frac{E(\xi)}{ \alpha} E_{1,k}(\alpha) \right)^3}
 \\
  \hat{m}_4  &=&  \frac{\lambda}{4 \alpha} E(\xi^4) \overline{E_{j,4}(\alpha)} - \frac{4 \lambda \mu}{3 \alpha} (1-e^{-\alpha \Delta}) E(\xi^3)\overline{E_{j,3}(\alpha) e^{-j \alpha \Delta}}+ \frac{4 \lambda^2}{3 \alpha^2} E(\xi^3)E(\xi)\overline{E_{j,3}(\alpha)E_{j,1}(\alpha)}\\
  &+&\frac{3}{4 \alpha^2} \left(\lambda E(\xi^2) +\sigma ^2 \right)^2 \overline{E^2_{j,2}}\\
   &+&  \frac{3}{ \alpha}\left(\lambda E(\xi^2) +\sigma ^2 \right)\overline{\left(\mu  e^{-\alpha j \Delta}(1-e^{-\alpha  \Delta})-\frac{\lambda E(\xi)}{ \alpha} E_{j,1}(\alpha) \right)^2}\\
   &+& \overline{\left[ \mu e^{-j \alpha \Delta}(1-e^{-\alpha \Delta})- \lambda   \frac{E(\xi)}{ \alpha} E_{j,1}(\alpha) \right]^4}
 \end{eqnarray*}
where $\overline{f_j} = \frac{1}{n} \sum_{j=1}^n f_j$.\\
Hence:
\begin{eqnarray*}
% \nonumber to remove numbering (before each equation)
 \overline{E_{j,1}(\alpha)} &=& (1-e^{-\alpha  \Delta})\overline{e^{-\alpha j \Delta}}\\
 \overline{E_{j,k}(\alpha)} &=& (1-e^{-k\alpha  \Delta})+ (1-e^{-\alpha  \Delta})^k (1-\overline{e^{-k \alpha j \Delta}})\\
 \overline{e^{-\alpha j \Delta} E_{j,k}(\alpha)}&=& (1-e^{-k\alpha  \Delta})\overline{e^{-\alpha j \Delta}}+ (1-e^{-\alpha  \Delta})^k \overline{e^{-(k+1) \alpha j \Delta}}\\
  \overline{E_{j,l}(\alpha) E_{j,k}(\alpha)} &=& (1-e^{-k\alpha  \Delta})(1-e^{-l\alpha  \Delta})+(-1)^l(1-e^{-\alpha  \Delta})^l (1-e^{-k\alpha  \Delta}) (1-\overline{e^{-l\alpha j \Delta}})\\
  &+& (-1)^k (1-e^{-\alpha  \Delta})^k (1-e^{-l\alpha  \Delta}) (1-\overline{e^{-k\alpha j \Delta}})\\
   &+& (-1)^{l+k} (1-e^{-\alpha  \Delta})^{l+k}(1-\overline{e^{-l\alpha j \Delta}}-\overline{e^{-k\alpha j \Delta}}+\overline{e^{-(l+k) \alpha j \Delta}})
\end{eqnarray*}
Higher moments can be computed in a similar way.
\begin{example}\textit{MRBSch}\\
In the case of the BSchMR model notice that $K_1=K_2=0$. The matching of the first three moments leads to the equations:
\begin{eqnarray*}
 \hat{m}_1&=& \mu (1-e^{-\alpha  \Delta})\overline{e^{-\alpha j \Delta}}\\
  \hat{m}_2 &=& -\frac{\sigma ^2}{2 \alpha}\overline{E_{j,2}(\alpha)} - \mu^2  (1-e^{-\alpha \Delta})^2 \overline{e^{-2 \alpha j \Delta}}\\
  \hat{m}_3  &=&    \frac{3 \mu}{2 \alpha} \sigma ^2  (1-e^{-\alpha \Delta}) \overline{e^{- \alpha j \Delta} E_{j,2}(\alpha)}
+     \mu^3(1-e^{-\alpha \Delta})^3  \overline{ e^{-3 \alpha j \Delta}}
 \end{eqnarray*}
\end{example}
\begin{example}\textit{MRMe}
\begin{eqnarray*}
 \hat{m}_1&=& (\mu - \lambda  \frac{\mu_J}{ \alpha})(1-e^{-\alpha  \Delta}) \overline{e^{-\alpha  \Delta}}\\
  \hat{m}_2 &=& -\frac{\overline{E_{j,2}(\alpha)}}{2 \alpha} \left(\lambda (\sigma^2_J+\mu^2_J) +\sigma ^2 \right) - \mu^2  (1-e^{-\alpha \Delta})^2 \overline{e^{-2 \alpha j \Delta}}\\
  &+& 2 \frac{\lambda \mu \mu_J}{ \alpha} \overline{e^{- \alpha j \Delta}E_{j,1}(\alpha)}-\frac{\lambda^2  (\mu_J)^2}{ \alpha^2} \overline{E^2_{j,1}(\alpha)}
  \end{eqnarray*}

  \begin{eqnarray*}
 \hat{m}_3  &=&     \frac{ \lambda}{3 \alpha} E(\xi^3) \overline{E_{j,3}(\alpha)}+  \frac{3 \mu}{2 \alpha} \left(\lambda \sigma^2_J+\mu^2_J +\sigma ^2 \right) (1-e^{-\alpha \Delta}) \overline{e^{-j \alpha \Delta} E_{j,2}(\alpha)}\\
 &-& \frac{3 \lambda}{ 2 \alpha^2} \mu_J\left(\lambda \sigma^2_J+\mu^2_J +\sigma ^2 \right) \overline{E_{j,2}(\alpha) E_{j,1}(\alpha)}+ \overline{ \left(   \mu e^{-j \alpha \Delta}(1-e^{-\alpha \Delta})-\lambda   \frac{\mu_J}{ \alpha} E_{1,k}(\alpha) \right)^3}
 \\
  \hat{m}_4  &=&  \frac{\lambda}{4 \alpha} E(\xi^4) \overline{E_{j,4}(\alpha)} - \frac{4 \lambda \mu}{3 \alpha} (1-e^{-\alpha \Delta}) E(\xi^3)\overline{E_{j,3}(\alpha) e^{-j \alpha \Delta}}+\frac{4 \lambda^2}{3 \alpha^2} E(\xi^3)\mu_J\overline{E_{j,3}(\alpha)E_{j,1}(\alpha)}\\
  &+&\frac{3}{4 \alpha^2} \left(\lambda \sigma^2_J+\mu^2_J +\sigma ^2 \right)^2 \overline{E^2_{j,2}}\\
   &+&  \frac{3}{ \alpha}\left(\lambda \sigma^2_J+\mu^2_J +\sigma ^2 \right)\overline{\left(\mu  e^{-\alpha j \Delta}(1-e^{-\alpha  \Delta})-\frac{\lambda \mu_J}{ \alpha} E_{j,1}(\alpha) \right)^2}\\
   &+& \overline{\left[ \mu e^{-j \alpha \Delta}(1-e^{-\alpha \Delta})- \lambda   \frac{\mu_J}{ \alpha} E_{j,1}(\alpha) \right]^4}
 \end{eqnarray*}
\end{example}
\begin{example}\textit{MRKou}\\
The moments of the jump sizes are:
\begin{equation*}
  E(\xi^k)= k!\left( \frac{q}{\eta^k_1}-\frac{1-q}{\eta^k_2}\right)
\end{equation*}
Hence:
\begin{eqnarray*}
 \hat{m}_1&=& \mu \overline{e^{-\alpha  \Delta}}(1-e^{-\alpha  \Delta})- \lambda  \frac{1}{ \alpha}\left( \frac{q}{\eta^k_1}-\frac{1-q}{\eta^k_2}\right) \overline{E_{j,1}(\alpha)}\\
 &=& \left( \mu- \lambda  \frac{1}{ \alpha}\left( \frac{q}{\eta^k_1}-\frac{1-q}{\eta^k_2}\right) \right) (1-e^{-\alpha  \Delta}) \overline{e^{-\alpha  \Delta}}\\
  \hat{m}_2 &=& -\frac{\overline{E_{j,2}(\alpha)}}{2 \alpha} \left(2 \lambda \left( \frac{q}{\eta^2_1}-\frac{1-q}{\eta^2_2}\right) +\sigma ^2 \right) - \mu^2  (1-e^{-\alpha \Delta})^2 \overline{e^{-2 \alpha j \Delta}}\\
  &+& 2 \frac{\lambda \mu }{ \alpha}\left( \frac{q}{\eta_1}-\frac{1-q}{\eta_2}\right)\left(1-e^{-\alpha  \Delta} \right) \overline{e^{- 2 \alpha j }}-\frac{\lambda^2  \left( \frac{q}{\eta^k_1}-\frac{1-q}{\eta^k_2}\right)^2}{ \alpha^2} \overline{E^2_{j,1}(\alpha)}
  \end{eqnarray*}

  \begin{eqnarray*}
 \hat{m}_3  &=&     \frac{ 6 \lambda}{3 \alpha} \left( \frac{q}{\eta^3_1}-\frac{1-q}{\eta^3_2}\right) \overline{E_{j,3}(\alpha)}+  \frac{3 \mu}{2 \alpha} \left(2 \lambda \left( \frac{q}{\eta^2_1}-\frac{1-q}{\eta^2_2}\right) +\sigma ^2 \right) (1-e^{-\alpha \Delta}) \overline{e^{-j \alpha \Delta} E_{j,2}(\alpha)}\\
 &-& \frac{3 \lambda}{ 2 \alpha^2} \left( \frac{q}{\eta^k_1}-\frac{1-q}{\eta^k_2}\right) \left(2 \lambda \left( \frac{q}{\eta^2_1}-\frac{1-q}{\eta^2_2}\right) +\sigma ^2 \right) \overline{E_{j,2}(\alpha) E_{j,1}(\alpha)}\\
 &+&  \overline{ \left(   \mu e^{-j \alpha \Delta}(1-e^{-\alpha \Delta})-\lambda   \frac{1}{ \alpha} \left( \frac{q}{\eta_1}-\frac{1-q}{\eta_2}\right)E_{j,1}(\alpha) \right)^3}
 \\
  \hat{m}_4  &=&  \frac{6 \lambda}{ \alpha} \left( \frac{q}{\eta^4_1}-\frac{1-q}{\eta^4_2}\right) \overline{E_{j,4}(\alpha)} - \frac{4 \lambda \mu}{3 \alpha} (1-e^{-\alpha \Delta}) E(\xi^3)\overline{E_{j,3}(\alpha) e^{-j \alpha \Delta}}\\
  &+& \frac{12 \lambda^2}{ \alpha^2} \left( \frac{q}{\eta^3_1}-\frac{1-q}{\eta^3_2}\right)\left( \frac{q}{\eta_1}-\frac{1-q}{\eta_2}\right)\overline{E_{j,3}(\alpha)E_{j,1}(\alpha)}\\
  &+&\frac{3}{4 \alpha^2} \left(\lambda \left( \frac{q}{\eta^2_1}-\frac{1-q}{\eta^2_2}\right) +\sigma ^2 \right)^2 \overline{E^2_{j,2}}\\
   &+&  \frac{3}{ \alpha}\left(\lambda \left( \frac{q}{\eta^2_1}-\frac{1-q}{\eta^2_2}\right) +\sigma ^2 \right)\overline{\left(\mu  e^{-\alpha j \Delta}(1-e^{-\alpha  \Delta})-\frac{\lambda E(\xi)}{ \alpha} E_{j,1}(\alpha) \right)^2}\\
   &+& \overline{\left[ \mu e^{-j \alpha \Delta}(1-e^{-\alpha \Delta})- \lambda   \frac{1}{ \alpha}\left( \frac{q}{\eta_1}-\frac{1-q}{\eta_2}\right) E_{j,1}(\alpha) \right]^4}
 \end{eqnarray*}
\end{example}
\subsection{Estimation by Maximum Likelihood}
First, we find the p.d.f. of the random variables $X_{\Delta j}$. To this end we define the quantities $\gamma_t=\int_0^t e^{\alpha s} \;dB_s$ and $\nu_t=\int_0^t e^{\alpha s} \;dZ_s$.\\
From the jump-diffusion model given by equation (\ref{eq:bsjmultid3}) we can re-write equation (\ref{eq:mrevlogret}) as $ X_{j \Delta}= \beta_j+\eta_j$,  where the independent random variables  $\beta_j$ and $\eta_j$ are defined as:
\begin{eqnarray*}\nonumber
 \beta_j&=&  \mu (1-e^{-\alpha  \Delta})e^{-\alpha j \Delta} + \sigma e^{-\alpha (j+1) \Delta}(\gamma_{(j+1) \Delta}-\gamma_{j \Delta})+ \sigma e^{-\alpha j \Delta}(e^{-\alpha \Delta}-1 )\gamma_{j \Delta}\\
 \eta_j &=&  e^{-\alpha j \Delta} \left[ e^{-\alpha  \Delta} (\nu_{(j+1) \Delta}-\nu_{j \Delta})+ (e^{-\alpha  \Delta}-1)\nu_{j \Delta} \right]
 %&=&   e^{-\alpha j \Delta} \left[ e^{-\alpha  \Delta} \sum_{l=N_{j \Delta}+1}^{N_{(j+1) \Delta}}  e^{\alpha \xi_l}+ (e^{-\alpha  \Delta}-1)\sum_{l=1}^{N_{j \Delta}} e^{\alpha \xi_l} \right]
 \end{eqnarray*}
   From equation (\ref{eq:applevyk}) we have, that the characteristic functions for  $\gamma_{j \Delta}$ and $\gamma_{(j+1) \Delta}-\gamma_{j \Delta}$ are respectively:
 \begin{eqnarray*}
 % \nonumber to remove numbering (before each equation)
   \varphi_{\gamma_{j \Delta}}(u) &=& exp( \int_0^{j \Delta} l_B(i u e^{\alpha s}) \;ds)=exp(-\frac{1}{4 \alpha}(e^{2 \alpha  }-1)u^2)\\
    \varphi_{\gamma_{(j+1) \Delta}-\gamma_{j \Delta}}(u) &=&exp(-\frac{1}{4 \alpha}e^{2 \alpha j \Delta}(e^{2 \alpha  }-1)u^2)
    \end{eqnarray*}
 Therefore,  we conclude that:
    \begin{equation*}
    \beta_j \sim  N \left(\mu_{j,\beta}(\alpha), \sigma^2_{j,\beta} \right)
\end{equation*}
where $\mu_{j,\beta}(\alpha)=\mu (1-e^{-\alpha  \Delta})e^{-\alpha j \Delta}$ and $\sigma^2_{j,\beta}(\alpha)=\frac{\sigma^2}{4 \alpha} E_{j,2}(\alpha)$.\\
 On the other hand, from equation (\ref{eq:chfjumpdiff2}) the characteristic functions of $\nu_{j \Delta}$ and $\nu_{(j+1) \Delta}-\nu_{j \Delta}$ are respectively:
\begin{eqnarray*}\nonumber
     \varphi_{\nu_{j \Delta}}(u)&=& exp \left(  \int_0^{j \Delta} l_Z(i u e^{\alpha s}) \;ds \right)=exp \left( \lambda \int_0^{j \Delta} \varphi_{\xi}( u e^{\alpha s}) \;ds -\lambda j \Delta \right) \\ \nonumber
    \varphi_{\nu_{(j+1) \Delta}-\nu_{j \Delta}}(u)&=& exp \left( \lambda \int_{j \Delta}^{(j+1) \Delta} \varphi_{\xi}( u e^{\alpha s}) \;ds -\lambda  \Delta \right)
 \end{eqnarray*}
 Hence:
 \begin{eqnarray*}
    \varphi_{\eta_{j \Delta}}(u)  &=& \varphi_{\nu_{j \Delta}} \left((e^{-\alpha \Delta}-1 )e^{-\alpha j \Delta}u \right) \varphi_{\nu_{(j+1) \Delta}-\nu_{j \Delta}}(e^{-\alpha (j+1) \Delta}u)\\
    &=& exp \left( \lambda \int_0^{j \Delta} \varphi_{\xi}(e^{-\alpha (j+1) \Delta} u) \;ds -\lambda (j+1) \Delta \right)\\
    &&exp \left( \lambda \int_{j \Delta}^{(j+1) \Delta} \varphi_{\xi} \left((e^{-\alpha \Delta}-1 )e^{-\alpha j \Delta}u \right) \;ds  \right)\\
    &=& exp \left( \lambda(K_1(u)+K_2(u)-(j+1) \Delta)\right)
 \end{eqnarray*}
Notice that the probability distributions of $\nu_{j \Delta}$ and $\nu_{(j+1) \Delta}-\nu_{j \Delta}$ have positive mass probability at zero. We write their p.d.f.'s  as the Radon-Nikodym derivative with respect to a measure with positive mass at zero and diffuse everywhere else.\\
We denote by $f_{\beta_j}(x; \theta)$, $f_{\eta_j}(x; \theta)$ and $f_{X_{j \Delta}}(x; \theta)$ respectively the p.d.f.'s functions  of $\beta_j$, $\eta_j$ and $X_{j \Delta}$. In order to emphasize the dependence, we let them depend on  of the  unknown parameter $\theta$, which should not be confused with the Gerber-Shiu parameter in section one. We let other relevant quantities depend on $\theta$ as well.\\
 Furthermore, we assume the condition:
 \begin{eqnarray}\label{eq:condint}
    \int_{\mathbb{R}} exp(\lambda Re(K_1(u;\theta)+K_2(u;\theta)))\,du &<& +\infty
 \end{eqnarray}
 in order to guarantee the existence of the p.d.f. of $\eta_j$ and the log-return variables $X_{j \Delta}$.
 The  p.d.f. of the random variable $\eta_j$  is computed via inverse Fourier transform as:
 \begin{eqnarray}\nonumber
  f_{\eta_j}(x; \theta)  &=& \frac{1}{2 \pi} e^{-\lambda (j+1) \Delta}(1-exp(-\lambda(j+1) \Delta)) \int_{\mathbb{R}} exp(-iux+ \lambda (K_1(u; \theta)+K_2(u; \theta)))\,du \\ \nonumber
  &+& exp(-\lambda(j+1) \Delta) \\ \label{eq:pdfjump}
    &&
 \end{eqnarray}
  Notice that, by the independence between $\beta_j$ and $\eta_j$ we have $f_{X_{j \Delta}}= f_{\beta_j}\star f_{\eta_j} $, where $f \star g$ is the convolution product of functions $f$ and $g$.\\
Hence:
 \begin{eqnarray*}
    f_{X_{j \Delta}}(x ; \theta) &=&  \int_{\mathbb{R}}  f_{\beta_j}(x-y ; \theta) f_{\eta_j} (y ; \theta)\;dy\\
    &=&  \int_{\mathbb{R}}  f_{\eta_j} (x-\sigma_{j,\beta}(\alpha)z-\mu_{j,\beta}(\alpha))f_Z(z) \;dz
    \end{eqnarray*}
    after the change of variables $z=\frac{x- y-\mu_{j,\beta}(\alpha)}{\sigma_{j,\beta}(\alpha)}$. Then, we substitute equation (\ref{eq:pdfjump}) into the last equation above to get:
    \begin{eqnarray}\nonumber
    f_{X_{j \Delta}}(x;\theta)  &=&  \frac{1}{2 \pi} e^{-\lambda (j+1) \Delta}(1-exp(-\lambda(j+1) \Delta)\\ \nonumber
     && \int_{\mathbb{R}} \int_{\mathbb{R}} exp \left(-iu(x-\sigma_{j,\beta}(\alpha)z-\mu_{j,\beta}(\alpha))+ \lambda (K_1(u; \theta)+K_2(u; \theta))\right)\,du \, f_Z(z)\,dz \\ \nonumber
  &+& exp(-\lambda(j+1) \Delta)\\ \nonumber
  &=&  \frac{1}{2 \pi} e^{-\lambda (j+1) \Delta}(1-exp(-\lambda(j+1) \Delta)) J(x;\theta)+ exp(-\lambda(j+1) \Delta)\\ \label{eq:loglikeeta}
  &&
    \end{eqnarray}
    after applying Fubini, where:
    \begin{eqnarray*}
    % \nonumber to remove numbering (before each equation)
       J(x;\theta) &=& \int_{\mathbb{R}}  A_4(x,u;\theta) \,du\\
      A_4(x,u;\theta) &=&  \exp \left[-iu(x-\mu_{j,\beta}(\alpha))+ \lambda (K_1(u; \theta)+K_2(u; \theta))\right] \exp \left(-\frac{1}{2} \sigma^2_{j,\beta}(\alpha) u^2 \right)
    \end{eqnarray*}
  We denote the vector of data by $x_{\Delta}=(x_{\Delta}, x_{2\Delta}, \ldots, x_{ n \Delta})$. The log-likelihood function, disregarding terms non depending on the parameters and the last term, is:
 \begin{eqnarray}\nonumber
   l(x_{\Delta} ; \theta)&=& \sum_{j=1}^ n  log f_{X_{j \Delta}}(x_{j \Delta}; \theta)\\ \nonumber
   &=&- \lambda \Delta \left(\frac{n(n+3)}{2}\right)+ \sum_{j=1}^ n  \log  (1-\exp(-\lambda(j+1) \Delta)+ \sum_{j=1}^ n \log J(x_{j \Delta},\theta)\\ \label{eq:loglikejd}
   &&
 \end{eqnarray}
 The maximum likelihood estimator $\hat{\theta}_{MLE}:=arg min_{\theta}  l(x_{\Delta};\theta)$ solves the system:
 \begin{eqnarray*}
 % \nonumber to remove numbering (before each equation)
   D_k l(x_{\Delta} ; \theta) &=&  0,\;\; k=1,2,\ldots,d
 \end{eqnarray*}
 where $D_k$ here is the derivative with respect to the parameter $\theta_k$. The dimension $d$ depends on the specific model.\\
 Hence:
 \begin{eqnarray*}
 % \nonumber to remove numbering (before each equation)
   D_k l(x_{\Delta} ; \theta) &=&  \left \{\begin{array}{cc}
                                           -  \Delta \left(\frac{n(n+3)}{2}\right)+\Delta \sum_{j=1}^n  \frac{(j+1) )exp(-\lambda(j+1) \Delta)}{1-exp(-\lambda(j+1) \Delta)}+  \sum_{j=1}^n  \frac{D_k J(x_{j \Delta},\theta)}{J(x_{j \Delta},\theta)}= 0    &, \theta_k=\lambda  \\
                                         \sum_{j=1}^n  \frac{D_k J(x_{j \Delta},\theta)}{J(x_{j \Delta},\theta)}= 0 &, \theta_k \neq \lambda
                                           \end{array}
    \right.
 \end{eqnarray*}
   \begin{example}\textit{Mean-reverting Black-Scholes model}\\
In this case $\lambda=\xi_j=0$, $\theta=(\alpha,\mu,\sigma^2)$. Hence:
\begin{eqnarray*}
   l(x_{\Delta} ; \theta)&=& \frac{n}{2} \log \sigma^2 -\frac{n}{2} \log \alpha +\frac{1}{2}\sum_{j=1}^ n  \log E_{j,2}(\alpha)- \frac{2 \alpha}{\sigma^2}\sum_{j=1}^ n  \frac{(x_{j\Delta}-\mu_{j,\beta}(\alpha))^2}{E_{j,2}(\alpha)}
 \end{eqnarray*}
  Differentiating with respect to  $\alpha$:
 \begin{eqnarray*}
 % \nonumber to remove numbering (before each equation)
    \frac{\partial E_{j,2}(\alpha)}{\partial \alpha}&=&  2 \Delta e^{-2 \alpha \Delta}+\Delta (1-e^{- \alpha \Delta})(1-e^{- 2\alpha j \Delta})e^{- \alpha \Delta}+ 2 j \Delta (1-e^{- \alpha \Delta})^2 e^{- 2\alpha j \Delta}\\
    &=& \Delta \left[ e^{- 2 \alpha \Delta}-(4j+1)e^{-  \alpha (2j+1)\Delta}+(2j+1)e^{-  2 \alpha (j+1)\Delta}+2j e^{- 2 \alpha j \Delta}+e^{- \alpha \Delta}\right]
 \end{eqnarray*}
  we have the following system of equations:
 \begin{eqnarray*}
  \sigma^2 \frac{\partial l}{\partial \sigma^2}  &=& \frac{n}{2 }+ n \log \alpha +\frac{2 \alpha}{\sigma^2} \sum_{j=1}^ n \frac{(x_{j\Delta}-\mu (1-e^{-\alpha  \Delta})e^{-\alpha j \Delta})^2}{E_{j,2}(\alpha)}=0\\
 - \frac{\sigma^2}{4 \alpha (1-e^{-\alpha \Delta})}  \frac{\partial l}{\partial \mu}  &=& \sum_{j=1}^ n \frac{(x_{j\Delta}-\mu (1-e^{-\alpha  \Delta})e^{-\alpha j \Delta}) e^{-\alpha j \Delta}}{E_{j,2}(\alpha)}=0 \\
   &=& \frac{4 \alpha}{\sigma^2}(1-e^{-\alpha \Delta}) \left[\sum_{j=1}^ n \frac{x_{j\Delta}e^{-\alpha j \Delta}}{E_{j,2}(\alpha)}-
   \mu (1-e^{-\alpha \Delta})\sum_{j=1}^ n  \frac{e^{-\alpha j \Delta}}{E_{j,2}(\alpha)} \right]=0\\
   &=& \frac{4 \alpha n}{\sigma^2}(1-e^{-\alpha \Delta}) \left[ \overline{ \left(\frac{x_{j\Delta}e^{-\alpha j \Delta}}{E_{j,2}(\alpha)} \right)}-
    \mu (1-e^{-\alpha \Delta}) \overline{\left(\frac{e^{-\alpha j \Delta}}{E_{j,2}(\alpha)}\right)} \right]=0\\
  \frac{\sigma^2}{n}  \frac{\partial l}{\partial \alpha}  &=& -\frac{\sigma^2}{2 \alpha}+\frac{\sigma^2}{2}\overline{\left(\frac{\partial E_{j,2}(\alpha)}{\partial \alpha}\right)}
    - 2 \overline{\left( \frac{(x_{j\Delta}-\mu (1-e^{-\alpha  \Delta})e^{-\alpha j \Delta})^2}{E_{j,2}(\alpha)}\right)}\\
      &+& 4  \alpha \mu \Delta \left[\overline{\left(\frac{(j+1) x_{j\Delta} e^{-\alpha (j+1)  \Delta} }{E_{j,2}(\alpha)}\right)}-\overline{\left(\frac{j  e^{-\alpha j \Delta}}{E_{j,2}(\alpha)}\right)}\right]\\
   &-& 4  \alpha \mu^2 \Delta(1-e^{-\alpha  \Delta})\left[\overline{\left( \frac{ (j+1)  e^{-\alpha (2j+1)  \Delta}}{E_{j,2}(\alpha)}\right)}- \overline{\left( \frac{ j e^{-2 \alpha j \Delta}}{E_{j,2}(\alpha)}\right)} \right]\\
   &+& 4  \alpha \overline{\left(  \frac{\frac{\partial E_{j,2}(\alpha)}{\partial \alpha}(x_{j\Delta}-\mu (1-e^{-\alpha \Delta})e^{-\alpha j \Delta})^2}{E^2_{j,2}(\alpha)}\right)}=0
      \end{eqnarray*}
      Details are left to the appendix.
   \end{example}
 \begin{example}\textit{MRMe}\\
 The parameter vector is $\theta=(\alpha,\mu,\sigma^2, \lambda, \mu_J, \sigma^2_J)$ with:
 \begin{eqnarray*}
  K_1(u;\theta) &=&  \int_{j \Delta}^{(j+1) \Delta} \varphi_{\xi}(u e^{-\alpha((j+1)\Delta-s)})\;ds= \frac{1}{\alpha}\left[ A_1(u,0,(j+1) \Delta)-A_1(u,0,j \Delta)\right] \\
   K_2(u;\theta) &=& \int_0^{j \Delta}  \varphi_{\xi}(u (e^{-\alpha \Delta}-1)e^{-\alpha(j \Delta-s)})\;ds=\frac{1}{\alpha} A_1(u(1-e^{-\alpha \Delta}),0,j \Delta)\\
     J(x_{j \Delta},\theta)  &=&  \int_{\mathbb{R}} \exp \left(-iu(x-\mu_{j,\beta}(\alpha))+ \lambda (K_1(u; \theta)+K_2(u; \theta))\right) \exp \left(-\frac{1}{2} \sigma^2_{j,\beta}(\alpha) u^2 \right) \,du
   \end{eqnarray*}

 Hence equation (\ref{eq:loglikejd}) becomes:
 \begin{eqnarray*}
   l(x_{\Delta} ; \theta)&=& \sum_{j=1}^ n  log f_{X_{j \Delta}}(x_{j \Delta}; \theta)\\
   &=&- \lambda \Delta \left(\frac{n(n+3)}{2}\right)+ \sum_{j=1}^ n  \log  (1-\exp(-\lambda(j+1) \Delta)+ \sum_{j=1}^ n \log J(x_{j \Delta},\theta)\\
 \end{eqnarray*}

 Condition (\ref{eq:condint}) allows the interchange of derivative and integral leading to:
 \begin{eqnarray*}
   D_k J(x_{j \Delta},\theta)  &=&  \int_{\mathbb{R}} D_k \exp \left(-iu(x-\mu_{j,\beta}(\alpha))+ \lambda (K_1(u; \theta) \right. \\
   &+& \left. K_2(u; \theta))\right) \exp \left(-\frac{1}{2} \sigma^2_{j,\beta}(\alpha) u^2 \right) \,du
 \end{eqnarray*}
 \end{example}
\subsection{Generalized Estimation using Empirical Characteristic Function}
Empirical characteristic methods have been studied in several papers, see Yu(2003) and references within for an account of this approach.\\
The log-return US-bitcoin exchange rates assuming a mean-reverting Levy process enters within the framework of i.i.d., although non-stationary,  data.\\
We define the empirical characteristic function (ECF) as:
\begin{equation*}
  \hat{\varphi}_{X_{j\Delta}}(u)=\frac{1}{n} \sum_{j=1}^n exp(i u  X_{j \Delta})
\end{equation*}
The characteristic function is re-written $\varphi_{X_{j \Delta}}(u):=\varphi_{X_{j \Delta}}(u; \theta)$ to emphasize the dependence on the unknown parameter. \\
This parameter includes mean-reverting, diffusion and jump parameters, namely $\theta=(\alpha, \mu, \sigma, \lambda, \theta_{\xi}) \in \mathbb{R}^d$, where $\theta_X$ are the parameters related to the probability distribution of the jumps.\\
The estimation functions $f( X_{j \Delta}; \theta): \mathbb{R} \times \mathbb{R}^d \rightarrow \mathbb{R}^l$ and the estimating equations are defined as:
\begin{eqnarray*}
% \nonumber to remove numbering (before each equation)
  h( u, X_{j \Delta}; \theta) &=&  exp(i u  X_{j \Delta})-\varphi_{X_{j \Delta}}(u; \theta)\\
  f( X_{j \Delta}; \theta) &=& \left( Re h( u_1, X_{j \Delta}; \theta), \ldots Re h( u_L, X_{j \Delta}; \theta) \right.,\\
  && \left. Im  h( u_1, X_{j \Delta}; \theta), \ldots,  Im h( u_L, X_{j \Delta}; \theta) \right)\\
  \frac{1}{n} \sum_{j=1}^n  f( X_{j \Delta}; \theta)&=& 0
\end{eqnarray*}
where $u_k=-\eta+\delta k,\, k=1,2,\ldots,L$ is an equally spaced grid of length $\delta=\frac{2 \eta}{L}$ on the interval $[-\eta,\eta]$, where the estimating functions are evaluated. See Feuerverger and McDunnough(1981) for a discussion about the optimal choice of points $u_k$. \\
The GMM estimator $\hat{\theta}_{GMM}$ is obtained as:
\begin{equation*}
  \hat{\theta}_{GMM}=argmin_{\theta} \frac{1}{n} \sum_{j=1}^n  f( X_{j \Delta}; \theta) \hat{\Omega} \frac{1}{n} \sum_{j=1}^n  f( X_{j \Delta}; \theta)'
\end{equation*}
where $\hat{\Omega}$ is a consistent estimator of the matrix:
\begin{equation*}
  \Omega=\left(
           \begin{array}{cc}
             \Omega_{RR} & \Omega_{RI} \\
             \Omega_{IR} & \Omega_{II} \\
           \end{array}
         \right)
\end{equation*}
with components:
\begin{eqnarray*}
% \nonumber to remove numbering (before each equation)
 ( \Omega_{RR})_{jk}  &=& \frac{1}{2}(Re(\varphi_{X_{j \Delta}}(u_j+u_k; \theta))+Re(\varphi_{X_{j \Delta}}(u_j-u_k; \theta)))-Re(\varphi_{X_{j \Delta}}(u_j; \theta))Re(\varphi_{X_{j \Delta}}(u_k; \theta))\\
 ( \Omega_{RI})_{jk}  &=& \frac{1}{2}(Im(\varphi_{X_{j \Delta}}(u_j+u_k; \theta))+Im(\varphi_{X_{j \Delta}}(u_j-u_k; \theta)))-Im(\varphi_{X_{j \Delta}}(u_j; \theta))Re(\varphi_{X_{j \Delta}}(u_k; \theta))\\
 ( \Omega_{II})_{jk}  &=& \frac{1}{2}(Re(\varphi_{X_{j \Delta}}(u_j+u_k; \theta))+Re(\varphi_{X_{j \Delta}}(u_j-u_k; \theta)))-Im(\varphi_{X_{j \Delta}}(u_j; \theta))Im(\varphi_{X_{j \Delta}}(u_k; \theta))\\
 \Omega_{RI} &=& \Omega_{IR}
 \end{eqnarray*}
 A continuum choice of $u$, see Carrasco and Florens(2002) leads to the estimator $\theta_C$ verifying:
 \begin{equation*}
   \theta_C= arg min_{\theta} ||\hat{\varphi_{\Delta}}- \varphi_{X_{j \Delta}}(u; \theta)||_{W}
 \end{equation*}
 where $|| f||_W=\int_{\mathbb{R}} |f(u)|^2\, exp(-u)\,du$
 \begin{eqnarray}\nonumber
   \varphi_{X_{j \Delta}}(u)  &=&  exp \left[ \lambda  (K_1(u; \theta)+K_2(u; \theta)) \right.\\ \nonumber
  &-&\left. \lambda (j+1)\Delta +i \mu e^{-\alpha j \Delta}(1-{\mathrm{e}}^{-\alpha \Delta})u-\frac{\sigma ^2}{4 \alpha} E_{j,2}(\alpha)u^2 \right] \\ \label{eq:chfjumpdiffzero}
    &&
 \end{eqnarray}
 \section{Pricing  bitcoin options}
We study the  pricing of  a European call option. Its payoff is given by:
\begin{eqnarray*}
  h_1 &=& (S_T-K)_+:=max(S_T-K,0)
\end{eqnarray*}
To apply a FFT  method we redefine the payoff in terms of the log-returns instead of the  price of the exchange. Hence, we write:
\begin{equation}\label{payoffreturn}
  H(y)=(e^{y}-K)_+
\end{equation}
 We give the following basic result in terms of the pricing of a European contract by FFT inversion.  It is adapted from Car and Madan(1999) to these specific models. We introduce a damping factor $R$ for stability. See Raible(2001) for the latter.
\begin{proposition}\label{pricesatheo}
Consider a dynamic driven by equations (\ref{eq:priceslog}),(\ref{eq:bsjmultid4})and (\ref{eq:bsjmultid3}) under an EMM $\mathcal{Q}^{\theta}$ obtained by an Esscher transformation,  and a European call option with strike price $K$ and maturity $T$.\\
 Assume there exists  a real value $R>1$ such that $E_{\mathcal{Q}} [e^{R V_t}] < +\infty$.\\
Then,  the price of the contract is given by:
 \begin{eqnarray}  \label{eq:pricesa}
  C&:=& C(x_0)= \frac{1}{2 \pi} e^{R x_0-rT} \int_{\mathbb{R}}  e^{-i x_0 x}  \varphi^{\theta}_{Y_T}(-iR-x) \hat{H}_R(x)\;dx
    \end{eqnarray}
 where $x_0=\log S_0, k=\log(K)$ and
 \begin{eqnarray*}
    \hat{H}_R(x)&=&e^{(ix-R)(k-x_0)}K \left( \frac{1}{ix-R}-\frac{e^{-x_0}}{ix-R+1} \right)
 \end{eqnarray*}
  \end{proposition}
\begin{proof}
  We write $Y_T \sim \mathcal{Q}^{\theta}_{Y_T}(dx)$, where $Q^{\theta}_{Y_T}$ is the probability distribution of $Y_T$ under the EMM $\mathcal{Q}^{\theta}$.\\
 Denoting by   $H_R(x)= e^{-Rx}H(x) \in L^1(\mathbb{R})$   we have:
\begin{eqnarray*}
% \nonumber to remove numbering (before each equation)
 C(x_0) &=&  e^{-rT}E_{\mathcal{Q}^{\theta}}[H(Y_T+x_0)]=  e^{-r T} \int_{\mathbb{R}}H(y+x_0)Q^{\theta}_{Y_T}(dy)\\
 &=&  e^{-r T} \int_{\mathbb{R}}e^{R(y+x_0)}H_R(y+x_0)Q_{Y_T}(dy)\\
 &=& \frac{1}{2 \pi} e^{R x_0-r T} \int_{\mathbb{R}} e^{Ry} \left[\int_{\mathbb{R}}e^{-i(y+x_0+(r-m)(1-e^{-\alpha T})))x}\hat{H}_R(x)\; dx \right] Q^{\theta}_{Y_T}(dy)\\
 &=& \frac{1}{2 \pi} e^{R x_0-r T}\int_{\mathbb{R}} e^{-i (x_0+(r-m)(1-e^{-\alpha T})) x} \left[\int_{\mathbb{R}}e^{(R-ix)y} Q^{\theta}_{Y_T}(dy) \right] \hat{H}_R(x)\;dx\\
 &=&  \frac{1}{2 \pi} e^{R x_0-r T} \int_{\mathbb{R}} e^{-i x_0 x}  \varphi^{\theta}_{Y_T}(-iR-x) \hat{H}_R(x)\;dx
  \end{eqnarray*}
On the other hand:
\begin{eqnarray*}
\hat{H}_R(x) &=& \int_{\mathbb{R}} e^{i x y} H_R(y)\;dy=\int_{k-x_0}^{+\infty} e^{(ix-R)y}(e^y-e^{k})\;dy\\
  &=& \int_{k-x_0}^{+\infty} e^{(ix-R+1)y}\;dy-e^{k} \int_{k-x_0}^{+\infty} e^{(ix-R)y}\;dy\\
  &=&-\frac{e^{(ix-R+1)(k-x_0)}}{ix-R+1}+e^{k}\frac{e^{(ix-R)(k-x_0)}}{ix-R}
\end{eqnarray*}
\end{proof}
The integral in equation (\ref{eq:pricesa}) is efficiently calculate by a fast  FFT approach. To this end we define the grids:
\begin{eqnarray*}
  x_k &=& -M+ \eta k,\;k=0,1,\ldots,n-1\\
  x_{0,j} &=& x_{0,m}+\delta j,\; j=0,1,\ldots,n-1
\end{eqnarray*}
over the domains of the integration variable and the initial log-prices respectively. Here $\eta=\frac{2M}{n}$ and $\delta=\frac{\pi}{M}$ are their corresponding lengths, while $n$ is the number of points on both grids, typically a power of two\\
We  apply the trapezoid rule after truncating the integral on the interval $[-M,M]$:
\begin{eqnarray*}
% \nonumber to remove numbering (before each equation)
C(x_{0,j})&\simeq& \frac{1}{2 \pi} e^{R x_0-r T} \sum_{k=0}^{n-1} w_k e^{-i x_{0,j} x_k }\varphi^{\theta}_{Y_T}(-iR+M-\eta k) \hat{H}_R(-M+\eta k)\eta\\
&=& \frac{1}{2 \pi} e^{R x_0-r T} e^{i M (x_{0,m}+\delta j) }\sum_{k=0}^{n-1} w_k \varphi^{\theta}_{Y_T}(-iR+M-\eta k) \hat{H}_R(-M+\eta k)\eta e^{-i x_{0,m} \eta k }e^{-i \delta \eta jk }\\
&=& \frac{1}{2 \pi} e^{R x_0-r T} e^{i M (x_{0,m}+\delta j) }\sum_{k=0}^{n-1} w_k \varphi^{\theta}_{Y_T}(-iR+M-\eta k) \hat{H}_R(-M+\eta k)\eta e^{-i x_{0,m} \eta k }e^{-i \frac{2 \pi}{n} jk }\\
&=& \frac{1}{2 \pi} e^{R x_0-r T} e^{i M (x_{0,m}+\delta j)}\sum_{k=0}^{n-1} h_k e^{-i \frac{2 \pi}{n} jk }=\frac{1}{2 \pi} e^{R x_0-r T} e^{i M (x_{0,m}+\delta j)}fft(h_k)
\end{eqnarray*}
with:
\begin{equation*}
 h_k=  w_k \eta \varphi^{\theta}_{T_T}(-iR+M-\eta k) \hat{H}_R(-M+\eta k)
\end{equation*}
 $w_0=w_{n-1}=\frac{1}{2}$ and equal to one otherwise.\\
  The expression $fft(h_k)$ denotes the Fast fourier Transform of the sequence $(h_k)$.
  \section{Acknowledgments}
This research has been supported by the Natural Sciences and Engineering Research Council of Canada (NSERC).
\section{Conclusions}
We have introduced a model for the dynamic of the bitcoin-US dollar exchange that allows to capture important empirical features such as random jumps and mean-reverting properties. In addition, we have proposed a pricing method for European call options, adapting the well-known FFT approach for Levy processes. To this end we have given expressions for the characteristic function of log-returns of the exchanges and have established estimation procedures for the parameters in the model.
\section{Appendix}

\textbf{a) Moments for log-returns series}\\
\begin{eqnarray*}\nonumber
  D \varphi_{X_{j \Delta}}(u)  &=& D T_1(u) \varphi_{X_{j \Delta}}(u) \\
 D^2 \varphi_{X_{j \Delta}}(u)  &=& D^2 T_1(u) \varphi_{X_{j \Delta}}(u)+(D T_1(u))^2 \varphi_{X_{j \Delta}}(u)\\
 D^3 \varphi_{X_{j \Delta}}(u)  &=& D^3 T_1(u) \varphi_{X_{j \Delta}}(u)+3 D^2 T_1(u) D T_1(u)\varphi_{X_{j \Delta}}(u)+(D T_1(u))^3 \varphi_{X_{j \Delta}}(u)\\
 D^4 \varphi_{X_{j \Delta}}(u) &=& D^4 T_1(u) \varphi_{X_{j \Delta}}(u)+4 D^3 T_1(u) D T_1(u)\varphi_{X_{j \Delta}}(u)\\
 &+& 3(D^2 T_1(u))^2 \varphi_{X_{j \Delta}}(u)+ 6 (D T_1(u))^2 D^2 T_1(u)\varphi_{X_{j \Delta}}(u)+(D T_1(u))^4 \varphi_{X_{j \Delta}}(u)
 \end{eqnarray*}
 After evaluating at zero the first and second derivatives are:
  \begin{eqnarray*}
  D \varphi_{X_{j \Delta}}(0)  &=&  i \mu e^{-j \alpha \Delta}(1-e^{-\alpha \Delta})+\lambda (D K_1(0)+ D K_2(0))  \\
  D^2 \varphi_{X_{j \Delta}}(0)  &=& \lambda (D^2 K_1(0)+ D^2 K_2(0))  \\ \nonumber
  &-& \frac{\sigma ^2}{2 \alpha} (1-{\mathrm{e}}^{-2 \alpha  \Delta}+(1-{\mathrm{e}}^{-2 \alpha j \Delta })(1-e^{-\alpha \Delta})^2)  \\
& +& \left[ i \mu e^{-j \alpha \Delta}(1-e^{-\alpha \Delta})+\lambda (D K_1(0)+ D K_2(0))\right]^2
\end{eqnarray*}
while the third and forth ones are:
 \begin{eqnarray*}
 D^3 \varphi_{X_{j \Delta}}(0)  &=& \lambda (D^3 K_1(0)+ D^3 K_2(0)) +3 \left[\lambda (D^2 K_1(0)+ D^2 K_2(0)) \right. \\ \nonumber
  &-& \left. \sigma ^2 ((1-{\mathrm{e}}^{-2 \alpha(j+1)\Delta})+(1-{\mathrm{e}}^{-2 \alpha j \Delta })(e^{-\alpha \Delta}-1)^2) \right] \\
  && \left[ i \mu e^{-j \alpha \Delta}(1-e^{-\alpha \Delta})+\lambda (D K_1(0)+ D K_2(0)) \right] \\
 &+& \left[ i \mu e^{-j \alpha \Delta}(1-e^{-\alpha \Delta})+\lambda (D K_1(0)+ D K_2(0)) \right]^3
 \end{eqnarray*}
 \begin{eqnarray*}
  D^4 \varphi_{X_{j \Delta}}(0) &=& \lambda (D^4 K_1(0)+ D^4 K_2(0)) +4 \lambda (D^3 K_1(0)+ D^3 K_2(0))\\
  && \left[ i \mu e^{-j \alpha \Delta}(1-e^{-\alpha \Delta})+\lambda (D K_1(0)+ D K_2(0))  \right]\\
 &+& 3 \left(\lambda (D^2 K_1(0)+ D^2 K_2(0)) \right.  \\ \nonumber
  &-& \left.  \sigma ^2 ((1-{\mathrm{e}}^{-2 \alpha(j+1)\Delta})+(1-{\mathrm{e}}^{-2 \alpha j \Delta })(e^{-\alpha \Delta}-1)^2) \right)^2\\
   &+& 6 \left(i \mu e^{-j \alpha \Delta}(1-e^{-\alpha \Delta})+\lambda (D K_1(0)+ D K_2(0)) \right)^2 \left[\lambda (D^2 K_1(0)+ D^2 K_2(0)) \right. \\ \nonumber
  &-& \left. \sigma ^2 ((1-{\mathrm{e}}^{-2 \alpha(j+1)\Delta})+(1-{\mathrm{e}}^{-2 \alpha j \Delta })(e^{-\alpha \Delta}-1)^2) \right]\\
   &+&  \left[ i \mu e^{-j \alpha \Delta}(1-e^{-\alpha \Delta})+\lambda (D K_1(0)+ D K_2(0))  \right]^4
 \end{eqnarray*}
From proposition \ref{prop:logretchf} :
\begin{eqnarray*}
  E( X_{j \Delta})&=&  \mu e^{-\alpha j \Delta}(1-e^{-\alpha  \Delta})-i \lambda (D K_1(0)+ D K_2(0))\\
  E( X^2_{j \Delta}) &=& -\left[\lambda (D^2 K_1(0)+ D^2 K_2(0)) \right. \\
  &-& \left. \frac{\sigma ^2}{2 \alpha} (1-{\mathrm{e}}^{-2 \alpha  \Delta}+(1-{\mathrm{e}}^{-2 \alpha j \Delta })(1-e^{-\alpha \Delta})^2)  \right]\\
  &+& \left[ i \mu e^{-j \alpha \Delta}(1-e^{-\alpha \Delta})-i\lambda (D K_1(0)+ D K_2(0)) \right]^2
  \end{eqnarray*}

  \begin{eqnarray*}
  E( X^3_{j \Delta}) &=&   i \lambda (D^3 K_1(0)+ D^3 K_2(0))-3 \left[  \mu e^{-j \alpha \Delta}(1-e^{-\alpha \Delta})- i \lambda (D K_1(0)+ D K_2(0)) \right] \\
&&  \left[\lambda (D^2 K_1(0)+ D^2 K_2(0))- \frac{\sigma ^2}{2 \alpha} (1-{\mathrm{e}}^{-2 \alpha  \Delta}+(1-{\mathrm{e}}^{-2 \alpha j \Delta })(1-e^{-\alpha \Delta})^2)  \right]\\
  &+&  \left[  \mu e^{-j \alpha \Delta}(1-e^{-\alpha \Delta})-i\lambda (D K_1(0)+ D K_2(0)) \right]^3 \\
  E( X^4_{j \Delta}) &=& \lambda(D^4 K_1(0)+ D^4 K_2(0))+4 \lambda (D^3 K_1(0)+ D^3 K_2(0))\\
   && \left[ i \mu e^{-j \alpha \Delta}(1-e^{-\alpha \Delta})+\lambda (D K_1(0)+ D K_2(0)) \right]\\
  &+& 3 \left( \lambda (D^2 K_1(0)+ D^2 K_2(0))  \right. \\ \nonumber
  &-& \left.  \frac{\sigma ^2}{2 \alpha} (1-{\mathrm{e}}^{-2 \alpha  \Delta}+(1-{\mathrm{e}}^{-2 \alpha j \Delta })(1-e^{-\alpha \Delta})^2) \right)\\
  &+& 6  \left[\lambda (D^2 K_1(0)+ D^2 K_2(0)) \right. \\ \nonumber
  &-& \left.  \frac{\sigma ^2}{2 \alpha} (1-{\mathrm{e}}^{-2 \alpha  \Delta}+(1-{\mathrm{e}}^{-2 \alpha j \Delta })(1-e^{-\alpha \Delta})^2)  \right]\\
  && \left[ i \mu e^{-\alpha j \Delta}(1-e^{-\alpha  \Delta})+\lambda (D K_1(0)+ D K_2(0)) \right]^2\\
  &+& \left[i \mu e^{-\alpha j \Delta}(1-e^{-\alpha  \Delta})+\lambda(D K_1(0)+ D K_2(0)) \right]^4
 \end{eqnarray*}
 Combining with equation (\ref{eq:chfjumpdiffzero}), the first four moments of the log-returns are:
\begin{eqnarray*}
  E( X_{j \Delta})&=&  \frac{1}{i} D T_1(0)= \mu e^{-j \alpha \Delta}(1-e^{-\alpha \Delta})+  \lambda   \frac{E(\xi)}{ \alpha} E_{j,1}(\alpha) \\
  E( X^2_{j \Delta}) &=& D^2 T_1(0) +(D T_1(0))^2 \\
  &=& - \frac{E_{j,2}(\alpha)}{2 \alpha} \left(\lambda E(\xi^2) +\sigma ^2 \right) + \left(i \mu e^{-j \alpha \Delta}(1-e^{-\alpha \Delta})-i \lambda   \frac{E(\xi)}{ \alpha} E_{j,1} \right)^2\\
  &=& - \frac{E_{j,2}(\alpha)}{2 \alpha} \left(\lambda E(\xi^2) +\sigma ^2 \right) - \left( \mu e^{-j \alpha \Delta}(1-e^{-\alpha \Delta})- \lambda   \frac{E(\xi)}{ \alpha} E_{j,1} \right)^2
  \end{eqnarray*}

  \begin{eqnarray*}
 E( X^3_{j \Delta}) &=& i(D^3 T_1(0)+3 D^2 T_1(0) D T_1(0)+(D T_1(0))^3) \\
 &=& i \left[i^3 \lambda  \frac{E(\xi^3)}{3 \alpha} E_{j,3}(\alpha)-  \frac{3 E_{j,2}(\alpha)}{2 \alpha} \left(\lambda E(\xi^2) +\sigma ^2 \right)\left(i \mu e^{-j \alpha \Delta}(1-e^{-\alpha \Delta})-i \lambda   \frac{E(\xi)}{ \alpha} E_{j,1}(\alpha) \right) \right.\\
 &+& \left. \left(  i \mu e^{-j \alpha \Delta}(1-e^{-\alpha \Delta})-i \lambda   \frac{E(\xi)}{ \alpha} E_{1,k}(\alpha) \right)^3 \right]\\
 &=&\lambda  \frac{E(\xi^3)}{3 \alpha} E_{j,3}(\alpha)+  \frac{3 E_{j,2}(\alpha)}{2 \alpha} \left(\lambda E(\xi^2) +\sigma ^2 \right)\left( \mu e^{-j \alpha \Delta}(1-e^{-\alpha \Delta})- \lambda   \frac{E(\xi)}{ \alpha} E_{j,1}(\alpha) \right) \\
 &+&  \left(  \mu e^{-j \alpha \Delta}(1-e^{-\alpha \Delta})- \lambda   \frac{E(\xi)}{ \alpha} E_{1,k}(\alpha) \right)^3
 \end{eqnarray*}

 \begin{eqnarray*}
  E( X^4_{j \Delta})  &=&  \lambda \frac{E(\xi^4)}{4 \alpha} E_{j,4}(\alpha) -  \frac{4 i^3 \lambda E(\xi^3)}{3 \alpha} E_{j,3}\left(i \mu e^{-j \alpha \Delta}(1-e^{-\alpha \Delta})-i \lambda   \frac{E(\xi)}{ \alpha} E_{j,1}(\alpha) \right)\\
 &+& \frac{3}{4 \alpha^2} \left(\lambda E(\xi^2) + \sigma ^2 \right)^2E^2_{j,2}(\alpha) \\
 &-&  \frac{3}{ \alpha} E_{j,2} \left(i \mu e^{-j \alpha \Delta}(1-e^{-\alpha \Delta})-i \lambda   \frac{E(\xi)}{ \alpha} E_{j,1} \right)^2  \left(\lambda E(\xi^2) + \sigma ^2 \right)\\
 &+& \left[i \mu e^{-j \alpha \Delta}(1-e^{-\alpha \Delta})-i \lambda   \frac{E(\xi)}{ \alpha} E_{j,1}(\alpha) \right]^4\\
  &=&  \lambda \frac{E(\xi^4)}{4 \alpha} E_{j,4}(\alpha) -  \frac{4  \lambda E(\xi^3)}{3 \alpha} E_{j,3}\left( \mu e^{-j \alpha \Delta}(1-e^{-\alpha \Delta})- \lambda   \frac{E(\xi)}{ \alpha} E_{j,1}(\alpha) \right)\\
 &+& \frac{3}{4 \alpha^2} \left(\lambda E(\xi^2) + \sigma ^2 \right)^2E^2_{j,2}(\alpha) \\
 &+&  \frac{3}{ \alpha} E_{j,2} \left( \mu e^{-j \alpha \Delta}(1-e^{-\alpha \Delta})- \lambda   \frac{E(\xi)}{ \alpha} E_{j,1} \right)^2  \left(\lambda E(\xi^2) + \sigma ^2 \right)\\
 &+& \left[ \mu e^{-j \alpha \Delta}(1-e^{-\alpha \Delta})- \lambda   \frac{E(\xi)}{ \alpha} E_{j,1}(\alpha) \right]^4
 \end{eqnarray*}
\textbf{b) Maximum likelihood equations for a MRBSch model}
  \begin{eqnarray*}
       \frac{\partial l}{\partial \alpha}  &=&  -\frac{n}{2 \alpha}+\frac{1}{2}\sum_{j=1}^ n  \frac{\partial E_{j,2}(\alpha)}{\partial \alpha}- \frac{2 }{\sigma^2}\sum_{j=1}^ n  \frac{(x_{j\Delta}-\mu_{j,\beta}(\alpha))^2}{E^2_{j,2}(\alpha)}\\
   &-& \frac{2 \alpha}{\sigma^2}\left[-2\sum_{j=1}^ n  \frac{(x_{j\Delta}-\mu_{j,\beta}(\alpha))\frac{\partial \mu_{j,\beta}(\alpha)}{\partial \alpha}}{E_{j,2}(\alpha)}-\sum_{j=1}^ n  \frac{\frac{\partial E_{j,2}(\alpha)}{\partial \alpha}(x_{j\Delta}-\mu_{j,\beta}(\alpha))^2}{E^2_{j,2}(\alpha)} \right]
   \end{eqnarray*}
   After elementary algebraic manipulations:
   \begin{eqnarray*}
       \frac{\partial l}{\partial \alpha} &=&-\frac{n}{2 \alpha}+\frac{1}{2}\sum_{j=1}^ n  \frac{\partial E_{j,2}(\alpha)}{\partial \alpha}- \frac{2 }{\sigma^2}\sum_{j=1}^ n  \frac{(x_{j\Delta}-\mu_{j,\beta}(\alpha))^2}{E^2_{j,2}(\alpha)}\\
         &+& \frac{4 \alpha}{\sigma^2}\sum_{j=1}^ n  \frac{(x_{j\Delta}-\mu_{j,\beta}(\alpha))\frac{\partial \mu_{j,\beta}(\alpha)}{\partial \alpha}}{E_{j,2}(\alpha)}+ \frac{4 \alpha}{\sigma^2}\sum_{j=1}^ n  \frac{\frac{\partial E_{j,2}(\alpha)}{\partial \alpha}(x_{j\Delta}-\mu_{j,\beta}(\alpha))^2}{E^2_{j,2}(\alpha)} \\
       &=&-\frac{n}{2 \alpha}+\frac{1}{2}\sum_{j=1}^n  \frac{\partial E_{j,2}(\alpha)}{\partial \alpha}- \frac{2}{\sigma^2}\sum_{j=1}^n  \frac{(x_{j\Delta}-\mu (1-e^{-\alpha  \Delta})e^{-\alpha j \Delta})^2}{E_{j,2}(\alpha)}\\
   &+& \frac{4 \alpha}{\sigma^2}\sum_{j=1}^n  \frac{(x_{j\Delta}-\mu (1-e^{-\alpha  \Delta})e^{-\alpha j \Delta})\mu((j+1) \Delta e^{-\alpha (j+1)  \Delta} -j \Delta e^{-\alpha j \Delta})}{E_{j,2}(\alpha)}\\
   &+& \frac{4 \alpha}{\sigma^2}\sum_{j=1}^ n  \frac{\frac{\partial E_{j,2}(\alpha)}{\partial \alpha}(x_{j\Delta}-\mu (1-e^{-\alpha \Delta})e^{-\alpha j \Delta})^2}{E^2_{j,2}(\alpha)}
    \end{eqnarray*}
    Hence:
     \begin{eqnarray*}
       \frac{\partial l}{\partial \alpha} &=&-\frac{n}{2 \alpha}+\frac{1}{2}\sum_{j=1}^n  \frac{\partial E_{j,2}(\alpha)}{\partial \alpha}- \frac{2}{\sigma^2}\sum_{j=1}^n  \frac{(x_{j\Delta}-\mu (1-e^{-\alpha  \Delta})e^{-\alpha j \Delta})^2}{E_{j,2}(\alpha)}\\
      &+& \frac{4 \alpha \mu \Delta}{\sigma^2}\sum_{j=1}^n  \frac{x_{j\Delta}((j+1)  e^{-\alpha (j+1)  \Delta} -j  e^{-\alpha j \Delta})}{E_{j,2}(\alpha)}\\
   &-& \frac{4 \alpha \mu^2 \Delta}{\sigma^2}(1-e^{-\alpha  \Delta})\sum_{j=1}^n  \frac{ ((j+1)  e^{-\alpha (2j+1)  \Delta} -j e^{-2 \alpha j \Delta})}{E_{j,2}(\alpha)}\\
   &+& \frac{4 \alpha}{\sigma^2}\sum_{j=1}^ n  \frac{\frac{\partial E_{j,2}(\alpha)}{\partial \alpha}(x_{j\Delta}-\mu (1-e^{-\alpha \Delta})e^{-\alpha j \Delta})^2}{E^2_{j,2}(\alpha)}\\
    &=&-\frac{n}{2 \alpha}+\frac{n}{2}\overline{\left(\frac{\partial E_{j,2}(\alpha)}{\partial \alpha}\right)}
    - \frac{2n}{\sigma^2} \overline{\left( \frac{(x_{j\Delta}-\mu (1-e^{-\alpha  \Delta})e^{-\alpha j \Delta})^2}{E_{j,2}(\alpha)}\right)}\\
      &+& \frac{4 n \alpha \mu \Delta}{\sigma^2}\left[\overline{\left(\frac{(j+1) x_{j\Delta} e^{-\alpha (j+1)  \Delta} }{E_{j,2}(\alpha)}\right)}-\overline{\left(\frac{j  e^{-\alpha j \Delta}}{E_{j,2}(\alpha)}\right)}\right]\\
   &-& \frac{4 n \alpha \mu^2 \Delta}{\sigma^2}(1-e^{-\alpha  \Delta})\left[\overline{\left( \frac{ (j+1)  e^{-\alpha (2j+1)  \Delta}}{E_{j,2}(\alpha)}\right)}- \overline{\left( \frac{ j e^{-2 \alpha j \Delta}}{E_{j,2}(\alpha)}\right)} \right]\\
   &+& \frac{4 n \alpha}{\sigma^2}\overline{\left(  \frac{\frac{\partial E_{j,2}(\alpha)}{\partial \alpha}(x_{j\Delta}-\mu (1-e^{-\alpha \Delta})e^{-\alpha j \Delta})^2}{E^2_{j,2}(\alpha)}\right)}
      \end{eqnarray*}

\end{document}